\PassOptionsToPackage{hidelinks}{hyperref}
\documentclass[manuscript]{acmart}

\usepackage{balance}
\usepackage{comment}
\usepackage{amsmath,amssymb,amsfonts}
\usepackage{algorithmic}
\usepackage{multicol}
\usepackage{multirow}
\usepackage{soul}
\usepackage{braket}
\usepackage{physics}
\usepackage[thinc]{esdiff}
\usepackage[]{mdframed}
\usepackage{CJKutf8}
\def\BibTeX{{\rm B\kern-.05em{\sc i\kern-.025em b}\kern-.08em
    T\kern-.1667em\lower.7ex\hbox{E}\kern-.125emX}}

\usepackage[usestackEOL]{stackengine}
\edef\tmp{\the\baselineskip}
\setstackgap{L}{\tmp}

\AtBeginDocument{%
  \providecommand\BibTeX{{%
    \normalfont B\kern-0.5em{\scshape i\kern-0.25em b}\kern-0.8em\TeX}}}

\setcopyright{acmcopyright}
\copyrightyear{2025}
\acmDOI{10.1145/1122445.1122456}

\acmBooktitle{~~}
\acmPrice{~~}
\acmISBN{~~}

\newtheorem{definition}{Definition}
\newtheorem{theorem}{Theorem}
\newtheorem{corollary}{Corollary}
\newtheorem{lemma}{Lemma}
\newtheorem{example}{Example}%
\newtheorem{remark}{Remark}%

\def\qed{\hfill\rule{2mm}{2mm}}

\usepackage{glossaries}
\newacronym{ai}{AI}{Artificial Intelligence}
\newacronym{xai}{xAI}{eXplainable Artificial Intelligence}
\newacronym{ml}{ML}{Machine Learning}
\newacronym{qml}{QML}{Quantum ML}
\newacronym{xqml}{XQML}{eXplainable Quantum ML}
\newacronym{xqai}{XQAI}{eXplainable Quantum AI}
\newacronym{nserc}{NSERC}{Natural Sciences and Engineering Research Council}
\newacronym{nist}{NIST}{National Institute of Standards and Technology}
\newacronym{shap}{SHAP}{SHapley Additive exPlanation}
\newacronym{lime}{LIME}{Local Interpretable
Model-agnostic Explanations}
\newacronym{gdpr}{GDPR}{General Data Protection Regulation}

\begin{document}

\title{A Shapley Value Estimation Speedup for Efficient Explainable Quantum AI}

\author{Iain Burge}
\authornotemark[1]
\email{iain-james.burge@telecom-sudparis.eu}
\affiliation{
  \institution{{\small SAMOVAR, T\'el\'ecom SudParis, Institut Polytechnique de Paris}}
%  \streetaddress{19 place Marguerite Perey}
%  \city{Palaiseau}
%%  \state{IdF}
  \country{91120 Palaiseau, France}
%  \postcode{91120}
}
\author{Michel Barbeau}
\authornotemark[2]
\email{barbeau@scs.carleron.ca}
\affiliation{%
  \institution{School of Computer Science, Carleton University}
  \streetaddress{1125 Col. By Dr.}
  \city{Ottawa}
  \state{Ontario}
  \country{Canada}
  \postcode{K1S 5B6}
}

\author{Joaquin Garcia-Alfaro}
\authornotemark[1]
\email{joaquin.garcia\_alfaro@telecom-sudparis.eu}
\affiliation{
  \institution{{\small SAMOVAR, T\'el\'ecom SudParis, Institut Polytechnique de Paris}}
  \streetaddress{19 place Marguerite Perey}
  \city{Palaiseau}
%  \state{IdF}
  \country{91120 Palaiseau, France}
%  \postcode{91120}
}

\renewcommand{\shortauthors}{Burge, Barbeau, Garcia-Alfaro}
\renewcommand{\shorttitle}{Quantum Shapley Value Estimation Speedup}

\begin{abstract}
This work focuses on developing efficient post-hoc explanations for quantum AI algorithms. 
In classical contexts, the cooperative game theory concept of the Shapley value adapts naturally to post-hoc explanations, where it can be used to identify which factors are important in an AI's decision-making process. 
An interesting question is how to translate Shapley values to the quantum setting and whether quantum effects could be used to accelerate their calculation.
We propose quantum algorithms that can extract Shapley values within some confidence interval. 
Our method is capable of quadratically outperforming classical Monte Carlo approaches to approximating Shapley values up to polylogarithmic factors in various circumstances.
We demonstrate the validity of our approach empirically with specific voting games and provide rigorous proofs of performance for general cooperative games.
\end{abstract}




\maketitle

\section{Introduction}

\noindent As \gls*{ai} becomes a larger part of critical decision-making processes, it is important to understand the logic behind the decisions being made.
Transparency in AI has become a topic of substantial regulatory importance worldwide.
In the European Union, the \gls*{gdpr} provides citizens the right to explanations for impactful automated decisions which relate to personal data~\cite{goodman2017european}. 
More recently, in 2024, the European Union enacted the AI act. 
The AI act provides individuals, in the context of high-risk AI systems, the right to an explanation for: (i) the use of an AI system in the decision-making process; (ii) the most important elements of that decision~\cite{nisevic2024explainable}.
In the United States, the \emph{Maintaining American Leadership in AI} executive order tasked the \gls*{nist} with developing a plan for robust and safe research and development in \gls*{ai} \cite{nannini2023explainability}.
\gls*{nist}'s plan listed explainability as an aspect of trustability, which is one of their key areas of focus.
This wave of legislative attention poses a substantial challenge, as many of today’s state-of-the-art \gls*{ai} algorithms, such as deep learning models, are unexplainable black boxes~\cite{rudin2019stop}.
Without specialized tools, AI developers often cannot understand the reasoning of their models.
There are several paths one can take to satisfy the new need for model explanations, the most popular approach is to create post hoc explanations for black-box models.
Post hoc explanations have the advantage that we can continue to use powerful black box models, such as computer vision models, without being entirely blind to their inner workings.

This paper\footnote{The work presented in this paper extends early research \cite{burgeQCE2023}.} focuses on additive explanations, which indicate inputs with the largest effect on a particular output~\cite{lundberg2017unified}.
For example, if one were to apply for a bank loan and be rejected, an additive explanation would quantify the impact of features such as income, location, and debt on the application’s rejection.
When applied to \gls*{ai} models, these explanations describe which inputs are most important in making a particular decision.
Several methods have been proposed for generating additive explanations; however, only one satisfies the important properties of local accuracy, missingness, and consistency~\cite{lundberg2017unified}.
This method is based on the Shapley value, a solution concept from cooperative game theory often used in economic game theory.

In game theory, the Shapley value is a weighted average of the contribution provided by a player to every possible coalition of other players.
To apply Shapley values to the analysis of \gls*{ai} models, we simply consider each feature a player and interpret a player being included or excluded from a coalition as a feature being on or off.
Unfortunately, the direct calculation of Shapley values is NP-Hard. 
This means it takes an exponential number of operations~\cite{matsui2001np, prasad1990np}.
Outside of some special cases, random sampling is the only option for approximation~\cite{castro2009polynomial}.

In parallel to the growth in \gls*{ai}, we have seen an emergence of quantum computation and quantum \gls*{ai}~\cite{biamonte2017quantum}.
Due to the fickle nature of quantum information, the problem of explainability is amplified since measuring a quantum system destroys information.
As a result, many quantum algorithms act as ultra-black boxes, where their internal workings are impossible to comprehend or even fully measure.
Though the field of \gls*{xqai} is still emerging, 
there exists some work into finding the Shapley values of quantum algorithms ~\cite{heese2023explaining,burge2023quantum}, 
and a bit into other additive methods such as LIME~\cite{deshmukh2023explainable}.
Existing methods for Shapley-based explanations either rely on knowing the algorithm's structure as a quantum circuit or random sampling.
By Chebyshev’s inequality, random sampling is quadratic in complexity, twice the precision is four times the work~\cite{saw1984chebyshev}. 
Fortunately, it is often possible to do better in by leveraging quantum effects.

In this paper, we detail an efficient algorithm for calculating the Shapley values of the input qubits of a quantum circuit.
In some cases, our method provides a quadratic speedup, up to polylogarithmic factors, when compared to Monte-Carlo methods.
With a fixed likelihood for success and a fixed range of outputs, we can have an arbitrarily accurate Shapley value approximation where the quantum query complexity grows inversely to precision.

\medskip

The paper is organized as follows.
Section~\ref{section:background} provides background and preliminaries on Shapley values.
Section~\ref{sec:classicalApplications} is composed of Subsection~\ref{sec:WeightedVotingGame}, which introduces a guiding example problem, and Subsection~\ref{sec:XAI}, which provides a simple example of Shapley Values use in explainability.
Sections~\ref{sec:QuantumRepresentation} and~\ref{sec:algorithm} present our algorithms. 
The error and complexity analysis are covered in Appendix~\ref{appendix:errorAnalysis}. 
Section~\ref{sec:example} demonstrates the application of our methods to the problem discussed in section~\ref{sec:classicalApplications}.
Section \ref{section:qShapImproved} provides an improved version of our quantum algorithm for Shapley value approximation. 
Section~\ref{section:relwork} surveys related work.
Section~\ref{sec:conclusion} concludes the article.

\section{Background}
\label{section:background}

\noindent This section presents preliminaries on Shapley values, including the notation used, as defined in Table~\ref{tab:notation-shapley-values}, and formal definitions.

\begin{table}[!hptb]
\centering
\caption{Notation and symbols used in this section}\label{tab:notation-shapley-values}
\begin{tabular}{l l p{12cm}}
 \hline
  $G=(F,V)$ & : & Coalitional game, $F$ being a set of players and $V$ a function which assigns a value to each subset of players (Definition~\ref{def:coalitionalGame})\\
  $\Phi(G,i), \Phi(i)$ & : & Player $i$'s payoff, or Shapley value, which represents how important they are in the game $G$ (Definitions~\ref{def:payoffVec},~\ref{def:shapleyvalue})\\
  $\Phi^\pm(i)$ & : & $\Phi^-(i)$ is a weighted average of the value of each subset not including player $i$ (Definition~\ref{def:shapleyvalue})\\
                & : & $\Phi^+(i)$ is a weighted average of the value of each subset including player $i$ (Definition~\ref{def:shapleyvalue})\\
                & : & The difference $\Phi^+(i)-\Phi^-(i)$ is equal to the $i$th player's Shapley value $\Phi(i)$ (Definition~\ref{def:shapleyvalue})\\
  $V^\pm$ & : & $V^-$ is the value of a subset excluding player $i$ (Definition~\ref{def:shapleyvalue})\\
          & : & $V^+$ is the value of a subset including player $i$ (Definition~\ref{def:shapleyvalue})\\
  $\gamma(n,m)$ & : & The weights used to calculate $\Phi^\pm(i)$ (Definition~\ref{def:shapleyvalue})\\
 \hline
\end{tabular}
\end{table}

Cooperative game theory is the study of coalitional games. 
In this article, we are most interested in Shapley values. 
We now list some definitions and preliminaries,

\begin{definition}[Coalitional game] \label{def:coalitionalGame}
    A \emph{coalitional game} is a tuple $G=(F, V)$. 
    \hbox{$F = \{0, 1, ..., n\}$} is a set of $n+1$ players.
    $V: \mathcal{P} (F) \xrightarrow{} \mathbb{R}$ is a value function with $V(S) \in \mathbb{R}$ representing the value of a given coalition $S \subseteq F$, with the restriction that $V(\emptyset) = 0$.
\end{definition}

\begin{definition}[Payoff vector] \label{def:payoffVec}
    Given a game $G=(F, V)$, there exists a \emph{payoff vector} $\Phi(G)$ of length $n+1$. 
    Each element $\Phi(G,i)\in \mathbb{R}$ represents the utility of player $i\in F$. 
    The value function determines a payoff vector.
    Player $i$'s payoff value $\Phi(G,i)$ is determined by how $V(S)$, $S\subseteq F$, is affected by $i$'s inclusion or exclusion from $S$.
\end{definition}

There are a variety of solution concepts for constructing payoffs~\cite{aumann2010some}.
We focus on the Shapley solution concept, which returns a payoff vector~\cite{winter2002shapley}.
Each element of the payoff vector $\Phi(G,i)$, is called player $i$'s Shapley value.
Shapley values have various interpretations depending on the game being analyzed.

Shapley values are derived using one of several sets of axioms. 
We use the following four \cite{winter2002shapley}.
Suppose we have games $G=(F, V)$ and $G'=(F, V')$, and a payoff vector $\Phi(G)$, then:
\begin{enumerate}
    \item Efficiency: The sum of all utility is equal to that of the grand coalition (the coalition containing all players),
        \begin{equation*}
            \sum\limits_{i=0}^{n} \Phi(G,i) = V(F).
        \end{equation*}
        
    \item Equal Treatment: Players $i$, $j$ are said to be symmetrical if for all $S\subseteq F$, where $i,j\notin S$ we have that \hbox{$V(S\cup \{i\}) = V(S \cup \{j\})$}.
    If $i$ and $j$ are symmetric in $G$, then they are treated equally, $\Phi(G,i) = \Phi(G,j)$.
    
    \item Null Player: Consider a player $i \in F$, if for all $S\subseteq F$ such that $i\notin S$, we have $V(S) = V(S\cup \{i\})$, then $i$ is a null player.
    If $i$ is a null player, then $\Phi(G,i) = 0$.
    
    \item Additivity: If a player is in two games $G$ and $G'$, then the Shapley values of the two games are additive
        \begin{align*}
            \Phi(G+G',i) = \Phi(G,i) + \Phi(G',i)
        \end{align*}
    where a game $G+G'$ is defined as $(F, V+V')$, and $(V+V')(S) = V(S) + V'(S)$, $S \subseteq F$.
\end{enumerate}
These axioms lead to a single unique and intuitive division of utility~\cite{winter2002shapley}.
The values of the payoff vectors can be interpreted as the responsibility of the respective players for the final outcome~\cite{hart1989shapley}.
When player $i$ has a small payoff $\Phi(G,i)$, then player $i$ has a neutral impact on the final outcome.
When player $i$ has a large payoff, then player $i$ greatly impacts the final outcome.

\begin{definition}[Shapley value \cite{shapley1952value}]
    \label{def:shapleyvalue}
    Let $G=(F, V)$, for notational simplicity, we write $\Phi(G,i)$ as $\Phi(i)$.
    The Shapley value of the $i^{th}$ player is,
    \begin{equation}
        \label{eq:payoff}
        \Phi(i)=\Phi^+(i)-\Phi^-(i),
    \end{equation}
    where,
    \begin{align}
        \label{eq:phi+}
        \Phi^+(i) &= \sum\limits_{S\subseteq F\setminus \{i\}} \gamma(\abs{F\setminus \{i\}}, \abs{S}) V^+(S), &V^+(S) &= V(S\cup \{i\}), \text{ and,}\\
        \label{eq:phi-}
        \Phi^-(i) &= \sum\limits_{S\subseteq F\setminus \{i\}}\gamma(\abs{F \setminus \{i\}}, \abs{S}) V^-(S), &V^-(S) &= V(S).
    \end{align}
    Given $n=\lvert F \setminus \{i\} \rvert$,
    \begin{equation*}
        \gamma(n, m) = \frac{1}{{n \choose m} (n+1)}.
    \end{equation*}
\end{definition}
\begin{remark} \label{rem:altShap}
    The Shapley value is equivalently written as,
    \begin{equation*} \label{}
        \Phi(i) = \sum\limits_{S\subseteq F\setminus \{i\}}\gamma(\abs{F \setminus \{i\}}, \abs{S}) (V(S\cup \set{i})-V(S)). 
    \end{equation*}
\end{remark}

The Shapley value of $i$ is the expected marginal contribution to a random coalition $S \subseteq F \setminus \{i\}$, where the marginal contribution is equal to $V(S \cup \{i\})- V(S)$ \cite{hart1989shapley}. 
Each player's Shapley value can be interpreted as a weighted average of their contributions.
Where the weights, $\gamma(n,m)$, have an intuitive interpretation: the factor $1/{n \choose m}$ results in each possible size of $S$ having an equal impact on the final value.
Since, given $\lvert S \rvert = m$, there would be ${n \choose m}$ summands contributing to the final value. 
The multiplicand $1/(n+1)$ averages between the different sizes of $S$.\\

\begin{lemma} \label{lemma:gammaSumsToOne}
    We have that $\sum_{S\subseteq F\setminus \{i\}} \gamma(\abs{F\setminus \{i\}}, \abs{S})$ is equal to one.
\end{lemma}
\begin{proof}
    Let us define $H_m = \{S\in\mathcal{P}(F\setminus \{i\}): \abs{S}\ = m\}$.
    We can rewrite $\sum_{S\subseteq F\setminus \{i\}} \gamma(\abs{F\setminus \{i\}}, \abs{S})$ as,
    \begin{equation*}
        \sum\limits_{m=0}^{n} \sum\limits_{S\in H_m} \gamma(n, m).
    \end{equation*}
    Plugging in the definition for $\gamma$,
    \begin{equation*}
        \sum\limits_{m=0}^{n} \sum\limits_{S\in H_m} \frac{1}{{n \choose m} (n+1)} = \frac{1}{n+1} \sum\limits_{m=0}^{n} \frac{1}{{n \choose m}} \sum\limits_{S\in H_m} 1.
    \end{equation*}
    Since there are $n$ choose $m$ possible subsets of size $m$ and hence $n$ choose $m$ possible subsets in $H_m$, it follows that we have, 
    \begin{equation*}
        \frac{1}{n+1} \sum\limits_{m=0}^{n} \frac{1}{{n \choose m}} {n\choose m}.
    \end{equation*}
    Hence, the result holds.
\end{proof}

We develop algorithms for general games and demonstrate them for a monotonic game.
\begin{definition}[Monotonic game] \label{def:monotonic}
    A game is \emph{monotonic} if for all $S,H \subseteq F$, we have, $V(S\cup H) \geq V(S)$. Note that when a game is monotonic, every summand in Equation~\eqref{eq:payoff} is non-negative.
\end{definition}

A naive approach to finding the $i^\text{th}$ player's Shapley value is through direct calculation using the Shapley Equation~\eqref{eq:payoff}, 
completing the task in $\mathcal{O}(2^n)$ assessments of $V$.
For structured games, it may be possible to calculate Shapley values more efficiently.
Otherwise, another option is random sampling~\cite{castro2009polynomial}.
Substantial trade-offs exist in each case. 
We propose a quantum algorithm with some substantial advantages in the following sections.

\section{Classical Shapley Value Applications}
\label{sec:classicalApplications}

\subsection{Weighted Voting Games}
\label{sec:WeightedVotingGame}

\noindent Let us model a player's voting power as a weighted count of instances in which the player has the deciding vote.
The Shapley values correspond to voting power.
Three friends sit around a table.
They are deliberating a grave matter.
Should they get Chinese food for the second weekend in a row?
They decide to take a vote.
Alice just got a promotion at work.
To celebrate this, their friends agreed to give them three votes.
Bob, the youngest of the group, also had good news, an incredible mark on their latest assignment! 
Everyone decided Bob should get two votes.
Charley, who had nothing to celebrate, and who is generally disliked, gets one vote.
The group decides to go out for Chinese food if there are four \emph{yes} votes.
In the end, all the friends vote for Chinese food.

At the restaurant, they run into their friend David, a mathematician.
David is intrigued when they hear about their vote. David begins to wonder how much power each friend had in the vote.
The intuitive answer is that Alice had the most voting power, Bob had the second most, and Charley had the least.
However, David notices something strange. 
Bob does not seem to have a more meaningful influence than Charley.
There are no cases where Bob's two votes would do more than Charley's single vote.
David concludes, there must be a more nuanced answer.
David heard of cooperative game theory and Shapley values.
David might be able to answer the question.

How can Definition~\ref{def:shapleyvalue} be applied to David's problem?
Consider the game $G=(F,V)$.
The players are Alice~(0), Bob~(1), and Charley~(2) represented as the set $F=\{0,1,2\}$.
Denote each player's voting weight as $w_0=3$, $w_1=2$, and $w_2=1$.
Recall that the vote threshold is $q$ equal to four.
Then, for any subset $S\subseteq F$, we can define $V$ as: 
\begin{equation} \label{eq:voteGameV}
     V(S) =
    \begin{cases}
        1 &\quad \text{if } \sum\limits_{j \in S} w_j \geq q,\\
        0 &\quad \text{otherwise.}
    \end{cases}
\end{equation}
$V(S)$ is one when the sum of players' votes in $S$ reaches the threshold of four.
Otherwise, the vote fails and $V(S)$ is zero.
This is called a weighted voting game~\cite{matsui2001np}. One could easily add more players with arbitrary non-negative weights and thresholds. 
Note that weighted voting games fall into the family of monotonic games (Definition~\ref{def:monotonic}).

In this context, the terms in the Shapley value equation have an intuitive meaning.
Take player $i$, and consider a set $S\subseteq F\setminus \{i\}$.
If $V(S\cup \{i\})-V(S)=1$, then the $i^\text{th}$ player is a deciding vote for the set of players $S$.
Otherwise, player $i$ is not a deciding vote.

Thus, for weighted voting games, player $i$'s Shapley value represents a weighted count of how many times $i$ is a deciding vote.
We can work out the Shapley values by hand, noting that,
\begin{equation*}
\begin{matrix}
    V(\emptyset) = 0, & V(\{0 \})  = 0, & V(\{1 \})  = 0, & V(\{ 0, 1\})  = 1, \\
    V(\{2 \}) = 0, & V(\{ 0, 2\}) = 1, & V(\{ 1, 2\})  = 0, & V(\{ 0, 1, 2 \})  = 1
\end{matrix}
\end{equation*}
From this, we have:

\begin{footnotesize}
\begin{align*}
    \Phi(0) &= \sum\limits_{S \subseteq F \setminus \{i\}} \gamma(\lvert F \setminus \{i\} \rvert, \lvert S \rvert) \cdot (V(S\cup \{i\}) - V(S)) \\
    &= \gamma(2, 0)\cdot(V(\{0 \})-V(\emptyset)) + \gamma(2,1)\cdot(V(\{ 0, 1\})-V(\{1 \}))\\
      &\quad+ \gamma(2,1)\cdot(V(\{ 0, 2\})-V(\{2 \})) + \gamma(2,2)\cdot(V(\{ 0, 1, 2\}) -V(\{1, 2\})) \\
    &= 2 \cdot \gamma(2, 1) + \gamma(2,2) 
    = \frac{2}{3}
\end{align*}
\end{footnotesize}
This can be repeated to get $\Phi(1)$ and $\Phi(2)$ equal to $1/6$.

In the case of Alice, Bob, and Charley's voting game, it is trivial to calculate their respective Shapley values.
However, what if one hundred colleagues were choosing a venue for a party, all with different numbers of votes?
In that case, a direct calculation would take $2^{100}$ assessments of $V$!
For this more general case, we need to be clever.

\subsection{Explainability and Shapley Values of Binary Classifiers}
\label{sec:XAI}

\noindent Explainable AI can be broken into two categories, inherent explainability, and post-hoc explainability \cite{rudin2019stop}. Inherently explainable models rely on algorithms which are easy to interpret, such as small decision trees or linear models. 
Ideally, every application would use inherently explainable models~\cite{rudin2019stop}; however, contemporary models tend to be black boxes which are very large and non-linear, e.g., deep neural networks.
As a result, post-hoc methods, which attempt to explain black box decisions, have become an important form of harm-reduction.
A promising research direction looks at using counterfactuals, questions of the form ``\emph{what is the smallest change that can be made to the input to change the output}," as explanations \cite{byrne2019counterfactuals, guidotti2024counterfactual}.
This paper focuses primarily on additive explanations, which assign importance to each input of a model \cite{lundberg2017unified}.

There are multiple approaches to producing model explanations, one of the most promising being based on Shapley Values~\cite{lundberg2017unified}.
We introduce a simplified but rigorous method to leverage Shapley Values for explainability.
Suppose we have a binary classifier $C:\{0,1\}^{r\times r} \rightarrow \{0,1\}$, which classifies binary strings of length $r^2$.
This could, for example, represent a classifier which takes a $r$ by $r$ black and white medical scan and decides whether a patient has a cancerous tumour \cite{chen2023quantum}.
This is easily translated to a cooperative game.
We consider each input bit (respectively pixel) to be a player in $F=\{0,\cdots,r^2-1\}$.
Each binary string \hbox{$h= h_{r^2-1} \cdots h_1 h_0 \in \{0,1\}^{r\times r}$} represents a coalition $S_h$ where player $j$ is included if and only if $h_j$ is one (or equivalently, the $j$th pixel is white),
\begin{equation*}
    S_h = \{j : h_j = 1, j\in \mathbb{Z}_r \}.
\end{equation*}
We can then define our value function $V_C:\mathcal{P}(F)\rightarrow \mathbb{R}$ to correspond to the classifier $C$,
\begin{equation*}
    V_C(S_h) = C(h_{r^2-1} \cdots h_1 h_0).
\end{equation*}
Using this change in perspective we can now leverage tools from cooperative game theory.
In particular, we can now assign each input bit $h_j$ an importance, its Shapley value $\Phi_j$.
For our tumour classifier example, this would give us a heat map of where the classifier focuses on the most when it makes decisions.
This represents a global explanation, which gives a general idea of which features (a.k.a., players or pixels) are most important for decision-making generally.

We can also define a method which gives local explanations.
This gives an idea of which features were most important for a particular decision.
Suppose we want to understand why classifier decided that $C(x)=y$, where \hbox{$x=x_{r^2-1}\cdots x_1 x_0\in\{0,1\}^{r \times r}$} and $y\in \{0,1\}$.
We define a new value function, as follows: 
\begin{equation}\label{eq:localExplainability}
    V_{C,x}(S_h) = \frac{1}{2^{\abs{F\setminus S_h}}}\sum\limits_{Q\subseteq F\setminus S_h} 
    \left| V_C(S_x) - V_C\left(
        (S_x\cap S_h)\cup Q
    \right) \right|.
\end{equation}

For a given set $Q\subseteq F\setminus S_h$, player $k\in F$ is in $(S_x\cap S_h)\cup Q$ if either: $k\in Q$ which implies $k\notin S_h$; or if $k\in S_h$ and $S_x$.
In more simple terms, when $k$ is in $S_h$, its inclusion in $(S_x\cap S_h)\cup Q$ is based on its inclusion in $S_x$, otherwise, it is based on its inclusion in $Q$. We average the change in classification across every possible $Q\subseteq F\setminus S_h$.
In effect, our value function fixes a subset of features $S_h$ in $x$, and every feature in $F\setminus S_h$ is replaced with noise.
The Shapley values derived from $V_{C,x}$ describe the importance of features in the decision of $C(x)=y$. Note that local explanations could be described differently, but this particular definition is highly compatible with the quantum algorithm described in the following sections.

\section{Quantum Algorithm for Shapley Value Approximation}
\label{sec:QuantumRepresentation}

\begin{table}[!hptb]
\centering
\small
\caption{Notation and symbols used in sections~\ref{sec:QuantumRepresentation}, \ref{sec:algorithm} and~\ref{section:qShapImproved}} \label{tab:notation-algorithm}
\begin{tabular}{l l p{12cm}}
 \hline
  $V_\text{max}, V_\text{min}$ & : & Maximum and minimum coalition values in the range of $V$, Equations~\eqref{eq:Vmax},~\eqref{eq:Vmin} \\
  $S_h$ & : & Set of players encoded by integer $h$'s binary representation (Definition~\ref{def:S_h}) \\
  $\hat{V}^\pm(h)$ & : & Normalized value of $S_h$ for $V^-$, and of $S_h \cup \{i\}$ for $V^+$, Equations~\eqref{eq:Vhat},~\eqref{eq:Vplus&Minus}\\
  $B^\pm$ & : & Operation taking $h$ and encoding a scaled $\hat{V}^\pm(h)$ in the amplitude of an output bit, Eq.~\eqref{eq:Bmatrix} \\
  $\abs{x}_H$ & : & Hamming distance from $0$ of integer $x$'s binary representation, Equation~\eqref{eq:hammingDist} \\
  $\ket{\mu^\pm}$ & : & Quantum state encoding $\Phi^\pm(i)$ in the expected value of measuring the final bit, Equation~\eqref{eq:mu}\\
  Player Register & : & Player register \texttt{Pl} stores a superposition of all possible player coalitions \\
  Utility Register & : & Value of each coalition is encoded in the amplitudes of the utility register \texttt{Ut} qubit \\
 \hline
\end{tabular}
\end{table}

We represent the Shapley value calculation problem in the quantum context.
Consider an $n+1$ player game $G$ represented by the pair $(F,V)$, where $F=\{0,1,\dots,n\}$ and $V: \mathcal{P} (F) \xrightarrow{} \mathbb{R}$, with $V(\emptyset) = 0$.
Table~\ref{tab:notation-algorithm} gives an overview for the notation in this and the next section (Sections~\ref{sec:QuantumRepresentation} and~\ref{sec:algorithm}).
We define $V_\text{max}$ as an upper bound and $V_\text{min}$ as the lower bound of the value function,
\begin{align}
    \label{eq:Vmax}
    V_\text{max} &\geq \max\limits_{S \subseteq F} V(S),\text{ and},\\
    \label{eq:Vmin}
    V_\text{min} &\leq \min\limits_{S \subseteq F} V(S).
\end{align}

The goal is to implement a quantum version of $V(S)$ and apply it on a superposition representing all possible subsets $S \subseteq F\setminus\{i\}$ and $S\cup \{i\}$.
We first define two quantum registers, the player and utility registers.
The player register \texttt{Pl} represents player coalitions and requires $n$ qubits.
Meanwhile, the utility \texttt{Ut} register requires one qubit, and in its probability amplitude, represents the output of $V$ given a player coalition.

\begin{definition} \label{def:S_h}
    To encode coalitions in the player register, consider the selection binary sequence \hbox{$h=h_n \cdots h_{i+1} h_{i-1}$} $\cdots$ \hbox{$h_0 \in \{ 0, 1\}^{n}$}.
    Let $S_h$ be the set of all players $j\in F$ such that $h_j=1$.
    $S_h$ can represent any player coalition that excludes the player $i$ given the corresponding $h$.
\end{definition}

In the player register, the computational basis states represent different subsets of players, where the $j$th qubit represents the $j$th player.
In this encoding, the $j$th qubit being $\ket{1}$ represents $j$'s inclusion in the subset.
$\ket{0}$ represents its exclusion.
Thus, every possible subset of players excluding player $i$ has a corresponding basis state.
It follows that representing every subset of players simultaneously in a quantum superposition is possible.
The one-qubit utility register encodes the output of $V$.

Next, we define a function that maps $V$ to the interval $[0,1]$,
\begin{equation}\label{eq:Vhat}
    \hat{V}^\pm(h) := \frac{V^\pm\left(S_h\right)-V_\text{min}}{V_\text{max}-V_\text{min}},
\end{equation}
where,
\begin{equation}\label{eq:Vplus&Minus}
    V^+(h) := \frac{V\left(S_h \cup \{i\}\right)-V_{\min}}{V_\text{max}-V_{\min}} \quad \text{and} \quad
    V^-(h) := \frac{V\left(S_h\right)-V_{\min}}{V_\text{max}-V_{\min}}.
\end{equation}

We define the following $2^{n+1}$ by $2^{n+1}$ block diagonal matrix:
\begin{equation}\label{eq:Bmatrix}
    B^\pm = \bigoplus\limits_{h=0}^{2^n-1} 
    \begin{pmatrix}
    \sqrt{1-\gamma(n,\abs{h}_H) \cdot \hat{V}^\pm(h)} &   \sqrt{\gamma(n,\abs{h}_H) \cdot \hat{V}^\pm(h)}  \\
    \sqrt{\gamma(n,\abs{h}_H) \cdot \hat{V}^\pm(h)} & -\sqrt{1-\gamma(n,\abs{h}_H) \cdot \hat{V}^\pm(h)}  
    \end{pmatrix},
\end{equation}
where $\abs{h}_H$ is the Hamming distance between $h$ and $0$ in $n$ bits, or the number of ones in the binary representation of $h$.
Formally, if $x = x_{n-1} \cdots x_0 \in \{0,1\}^{n}$, then,
\begin{equation} \label{eq:hammingDist}
    \abs{x}_H := \abs{\left\{
        j : x_j = 1, j\in \{0,\dots,n-1\}
    \right\}}
\end{equation}
As a result, $\abs{h}_H$ is the number of players in the coalition $S_h$.

\begin{theorem}
The block diagonal matrix $B^\pm(n)$ is unitary.
\end{theorem}
\begin{proof}
It follows from the fact that each block
\begin{equation*}
    \begin{pmatrix}
    \sqrt{1-\gamma(n,\abs{h}_H) \cdot \hat{V}^\pm(h)} &   \sqrt{\gamma(n,\abs{h}_H) \cdot \hat{V}^\pm(h)}  \\
    \sqrt{\gamma(n,\abs{h}_H) \cdot \hat{V}^\pm(h)} & -\sqrt{1-\gamma(n,\abs{h}_H) \cdot \hat{V}^\pm(h)}  
    \end{pmatrix}
\end{equation*}
of $B^\pm$ is unitary.
\end{proof}

We construct the quantum system:
\begin{equation} \label{eq:mu}
    \ket{\mu^\pm} = B^\pm (H^{\otimes n}\otimes I) \ket{0}_\texttt{Pl}^{\otimes n} \ket{0}_\texttt{Ut}.
\end{equation}
where $\ket{0}^{\otimes n}_\texttt{Pl}$ represents the player register, and $\ket{0}_\texttt{Ut}$ represents the utility register in their initial states.

\begin{theorem} 
The expected value of the utility register of system $\ket{\mu^+}$ minus the expected value of the utility register of system $\ket{\mu^-}$ is,
\begin{equation*}
    \frac{\Phi(i)}{2^{n} (V_\text{max}-V_\text{min})}.
\end{equation*}
\end{theorem}
\begin{proof}
It follows from the fact that
\begin{equation*}
    (H^{\otimes n}\otimes I) \ket{0}_\texttt{Pl} \ket{0}_\texttt{Ut}
    =
    \sum_{h=0}^{2^n-1} \frac{1}{\sqrt{2^n}} \ket{h}_\texttt{Pl}\ket{0}_\texttt{Ut}
\end{equation*}
and the following sequence of equivalences, where $B_{k+1,k}$ is the element of $B$ at row $k+1$ column $k$:

\begin{align*}
    \sum_{h=0}^{2^n-1} \left( B_{h+1,h} \right)^2 
     & = \sum_{h=0}^{2^n-1} \phi(h,n) = \sum_{h=0}^{2^n-1} \gamma(n,\abs{h}_H) \cdot \hat{V}^\pm(h) \\
     & =\sum\limits_{S \subseteq F \setminus \{i\}} \gamma(\lvert F \setminus \{i\} \rvert, \lvert S \rvert) \cdot \left(\frac{ V^\pm(S) - V_\text{min} }{V_{max}-V_{min}}\right)
\end{align*}

\noindent Hence, the probability of measuring a one in the utility register of the quantum system $\ket{\mu^\pm}$ is

\begin{equation*}
    \frac{1}{2^{n} \cdot (V_\text{max}-V_\text{min})}
    \left( - V_\text{min} + \sum_{S \subseteq F \setminus \{i\}} 
    \gamma(\lvert F \setminus \{i\} \rvert, \lvert S \rvert) 
    \cdot V^\pm(S) \right) 
\end{equation*}

\noindent Thus, the expected value of $\bra{\mu^+}(I^{\otimes n} \otimes \ketbra{1})\ket{\mu^+} - \bra{\mu^-}(I^{\otimes n} \otimes \ketbra{1})\ket{\mu^-}$ is,
\begin{equation}
    \frac{1}{2^{n} \cdot (V_\text{max}-V_\text{min})}
    \left( (V_\text{min} - V_\text{min}) + \sum_{S \subseteq F \setminus \{i\}} 
    \gamma(\lvert F \setminus \{i\} \rvert, \lvert S \rvert) 
    \cdot (V^+(S)-V^-(S)) \right).
\end{equation}
Applying the definition of $V^+$ and $V^-$ gives the result.
\end{proof}

The Shapley value $\Phi(i)$ can be obtained by repeatedly creating the quantum systems $\ket{\mu^+}$ and $\ket{\mu^-}$,
measuring their last qubit, subtracting their averages, and finally multiplying by $2^{n} \cdot (V_\text{max}-V_\text{min})$.
However, this representation requires an exponential number of queries to $B^\pm$ to estimate.
Additionally, this naive implementation of $B^\pm$ requires an exponential number of terms to be implemented.
In the following section, we develop a more efficient solution for the Shapley value calculation of cooperative games. 

\section{Efficient Quantum Algorithm for Shapley Value Approximation}
\label{sec:algorithm}

Consider an $n+1$ player game $G$. Suppose we have a quantum representation of the function $\hat{V}^\pm(S)$, defined in Equation~\eqref{eq:Vhat}, which acts on two registers, a player register, and a utility register.
Additionally, we introduce a third register, called the partition register, which we use to generate the weights in the Shapley value sum.
Table~\ref{tab:notation-algorithm} gives an overview of some of the notation in this and the previous section (Sections~\ref{sec:QuantumRepresentation} and~\ref{sec:algorithm}), while Table~\ref{tab:notation-efficient-algorithm} provides an overview for the remainder of the section.

\begin{table}[!hptb]
\centering
\small
\caption{Notation and symbols used in sections~\ref{sec:algorithm} and~\ref{section:qShapImproved}} \label{tab:notation-efficient-algorithm}
\begin{tabular}{l l p{12cm}}
 \hline
  $\epsilon$ & : & Total error of the efficient Shapley value approximation algorithm (Section~\ref{sec:algorithm})\\
  Partition Register & : & Partition register \texttt{Pt} helps to prepare an amplitude distribution that corresponds to the $\gamma(n,m)$'s in the Shapley Equation~\eqref{eq:payoff}\\
  $U_V^\pm$ & : & $U_V^-$ is a quantum implementation of $\hat V^-$, which outputs to the utility register, Equation~\eqref{eq:U_V}\\
    & : & $U_V^+$ is a quantum implementation of $\hat V^+$, which outputs to the utility register, Equation~\eqref{eq:U_V}\\
  $\ket{V^\pm(h)}$ & : & $\ket{V^-(h)}$ is the result of applying $U_V^-$ to the utility register with input integer $h$, Equation~\eqref{eq:ketV}\\
    & : & $\ket{V^+(h)}$ is the result of applying $U_V^+$ to the utility register with input integer $h$, Equation~\eqref{eq:ketV}\\
  $b_{n,m}(x)$ & : & Proportional to the Binomial distribution where $x$ is probability of success, $n$ is the number of trials, and $m$ is the number of successes (Definition~\ref{def:betaFunction}) \\
  $\beta_{n,m}$ & : & Integral of $b_{n,m}(x)$ from $0$ to $1$ (Definition~\ref{def:betaFunction}), $\beta_{n,m}$ is precisely equal to $\gamma(n,m)$ (Theorem~\ref{thm:beta=gamma}) \\
  $P_\ell$ & : & Partition of the interval from $0$ to $1$, made up of $2^\ell$ points, Equation~\eqref{eq:P_l} \\
  $t_\ell(k)$ & : & Location of the $k$th point in the partition $P_\ell$, Equation~\eqref{eq:t_l(k)} \\
  $w_\ell(k)$ & : & Width of the $k$th subinterval of partition $P_\ell$, Equation~\eqref{eq:w} \\
  $D_\ell$ & : & Quantum algorithm which prepares a state with amplitudes $\alpha_k\approx w_\ell(k)$, Equation~\eqref{eq:D} \\
  $H_m$ & : & Set of binary strings $h$ of a fixed length (i.e., $n$) such that $\abs{h}_H=m$ (Lemma~\ref{lemma:hammingDistRearrange})\\
  $C_D(\ell)$ & : & Complexity of implementing $D_\ell$ (Definition~\ref{def:C_D}) \\
  $C_V(\epsilon')$ & : & Complexity of implementing $U_V^\pm$ with maximum error upper bounded by $\epsilon'$ (Definition~\ref{def:C_V})\\
  $R_j$ & : & Takes partition register as input and rotates the $j$th player qubit accordingly, Equation~\eqref{eq:R_j} \\
 \hline
\end{tabular}
\end{table}

Approximating the Shapley value with additive error bounded by $\epsilon$ with success probability at least $8/\pi^2$ consists of the following steps:
\newcounter{algorithmicSteps}
\begin{enumerate}
    \item Represent every possible subset of players in a quantum superposition, excluding player $i$, with probability amplitudes corresponding to weights $\gamma(n,m)$ in the Shapley Equation~\eqref{eq:payoff}.
    \item Perform the quantum version of $V^\pm$ controlled by the player registers, while also potentially using a zeroed auxiliary register, on the utility register.
    If $V^+$ (respectively $V^-$) is applied, then the expected value of measuring the utility register corresponds to $\Phi^+(i)$ (respectively $\Phi^-(i)$).
    \setcounter{algorithmicSteps}{\value{enumi}}
\end{enumerate}
From this point, it is possible to approximate the Shapley value $\Phi(i)$ if one has the expected value of measuring the utility register.
\begin{enumerate}
    \setcounter{enumi}{\value{algorithmicSteps}}
    \item To approximate the expected value of measuring the utility register while introducing $\epsilon'$ error, perform amplitude estimation with $\mathcal{O}((V_\text{max}-V_\text{min})/ \epsilon')$ iterations, where Step~1 is $\mathcal{A}$ and Step~2 is $V^\pm$ as defined in Montanaro \cite{montanaro2015quantum}.
\end{enumerate}
The final error is less than or equal to $\epsilon$ with a fixed predetermined probability.
Analysis for error propagation of each step is included in Appendix~\ref{appendix:errorAnalysis}. 

\begin{mdframed}
\textbf{Efficient Quantum Algorithm for Shapley Value Approximation} 

\medskip

\noindent \textbf{Input:} Game $G=(F,V)$, where $F$ is a set of $n+1$ players, a quantum implementation $U_V^\pm$ of the value function~$V^\pm$, a player $i$, and error bound $\epsilon$.

\medskip

\noindent \textbf{Begin} with $\ket{\psi_0}$, a quantum state with three registers, i.e.,
\begin{equation*}
    \ket{\psi_0} = \ket{0}_\texttt{Pt} \otimes \ket{0}_\texttt{Pl} \otimes \ket{0}_\texttt{Ut}.
\end{equation*}
where the partition register \texttt{Pt} has \hbox{$\ell=\mathcal{O}(\log((V_{\max}-V_{\min})\cdot\sqrt{n}/\epsilon))$} qubits, player register \texttt{Pl} is $n$ qubits, and utility register \texttt{Ut} is one qubit.
The partition register helps prepare an amplitude distribution corresponding to the $\gamma(n,m)$'s in the Shapley Equation~\eqref{eq:payoff}.
The player register stores a superposition of all possible player coalitions. 
A player is excluded or included from a coalition if the bit in the player register with position corresponding to their number, is $0$ or $1$.
The value of each coalition is encoded in the amplitudes of the Utility register qubit.

\medskip

\noindent \textbf{Step~1:} Prepare the partition register such that,
\begin{equation*}
    \ket{\psi_{1a}} = \sum\limits_{k=0}^{2^\ell-1} \sqrt{w_\ell(k)} \ket{k}_\texttt{Pt} \ket{0}_\texttt{Pl} \ket{0}_\texttt{Ut}.
\end{equation*}
where \hbox{$w_\ell(k)=\sin^2\left(\pi (k+1)/ 2^{\ell+1}\right)-\sin^2\left(\pi k/ 2^{\ell+1}\right)$}.
Then, apply the circuit $R_j$ from Figure~\ref{fig:controlledRotationCircuit}, described mathematically in Equation~\eqref{eq:R_j}, to each player qubit, $j\in\{0,1,\dots,i-1,i+1,\dots,n-1,n\}$.
This yields the state $\ket{\psi_{1b}}$ (Equations~\eqref{eq:psi2} and~\eqref{eq:psi2_again}) which gives a superposition of all player coalitions with probability amplitudes that correspond to $\gamma(n,m)$ in Equations~\eqref{eq:phi+} and~\eqref{eq:phi-}.

\medskip
\noindent \textbf{Step~2:} Apply the quantum implementation of the value function $U_V^+$ where player $i$ is included, that uses the player register as an input and outputs to the utility register's amplitudes.
This gives $\ket{\psi_2^+}$.

\medskip
\noindent \textbf{Step~3:} Next, use the quantum subroutine to estimate a function's mean with bounded outputs from Montanaro~\cite{montanaro2015quantum}.
This approximates the expected value of measuring the utility register for $\ket{\psi_2^+} \approx \Phi^+(i)$.

\medskip

\noindent \textbf{Repeat} Steps~$1$ to~$3$ using the quantum implementation of the value function $U_V^-$ where player $i$ is excluded, yielding an approximation for $\Phi^-(i)$. 
Finally, we approximate the Shapley value of player $i$, $\Phi(i)$, as the difference $\Phi^+(i)-\Phi^-(i)$.

\end{mdframed}

We now describe the details of the algorithm.
The goal is to efficiently approximate the Shapley value $\Phi(i)$ of a given player $i\in F$.
We individually find $\Phi^+(i)$ and $\Phi^-(i)$ to simplify the quantum circuits.
Suppose the quantum representation of the function $\hat{V}^\pm(h)$, defined in Equation~\eqref{eq:Vhat}, is given as:
\begin{equation} \label{eq:U_V}
    U_V^\pm \ket{h}_\texttt{Pl}\ket{0}_\texttt{Ut} := \ket{h}_\texttt{Pl}\ket{V^\pm(h)}_\texttt{Ut}
\end{equation}
where $\ket{h}$ is a vector in the computational basis (i.e., \hbox{$h\in \{0,1\}^n$}), and where,
\begin{equation} \label{eq:ketV}
    \ket{V^\pm(h)} := \sqrt{1-\hat{V}^\pm(h)} \ket{0} + \sqrt{\hat{V}^\pm(h)} \ket{1}
\end{equation}

A critical step for the algorithm is to approximate the weights of the Shapley value function.
These weights correspond perfectly to a modified beta function.
\begin{definition}[Modified Beta Function] \label{def:betaFunction}
Let $n\in \mathbb{N}$ and $m\in\{0,\ldots,n\}$. We define the beta function as:
\begin{align*}
\beta_{n,m} &= \int\limits_0^1 b_{n.m}(x)dx, \text{ where}\\
b_{n.m}(x) &= x^{m}(1-x)^{n-m}.
\end{align*}
\end{definition}

\begin{theorem} \label{thm:beta=gamma}
The beta function $\beta_{n,m}$ is equal to the Shapley weight function $\gamma(n,m)$, for $n\in\mathbb{N}$ and $m\in\{0,\ldots, n\}$.
\end{theorem}
\begin{proof}
    See Appendix~\ref{appendix:betaIsGamma}.
\end{proof}

The beta function and, by extension, the Shapley weights are approximated with Riemann sums.
The sums represent the area under the function $x^m(1-x)^{n-m}$ using partitions of the interval $[0,1]$.
We demonstrate the function $x^m(1-x)^{n-m}$ can be implemented efficiently on a quantum computer.
\begin{definition}[Riemann sum]
    \label{def:riemannSum}
    A \emph{Riemann sum \cite[page~276]{ross2013elementary}} of a function $f$, with respect to a partition $P=(t_0,\ldots,t_s)$ of the interval $[a,b]$, is an approximation of the integral of $f$ from $a$ to $b$ of the form:
    \begin{equation*}
        \sum\limits_{k=0}^{s-1} (t_{k+1} - t_k) \cdot f(x_k) 
    \end{equation*}
    where $t_{k+1}-t_k$ is the width of the subinterval, and $f(x_k)$ is the height, $x_k\in[t_k,t_{k+1}]$.
\end{definition}

We begin with the initial quantum state:
\begin{equation}
    \ket{\psi_0} = \ket{0}_\texttt{Pt}^{\otimes \ell} \otimes \ket{0}_\texttt{Pl}^{\otimes n} \otimes \ket{0}_\texttt{Ut},
\end{equation}
where \texttt{Pt}, \texttt{Pl}, and \texttt{Ut} denote the partition, player, and utility registers.
We leverage \texttt{Pt} to prepare the $\gamma(n,m)$ weights.
Let $\ell\in \mathbb{N}$ be the number of qubits in \texttt{Pt}, with $\ell=\lceil \log_2((V_\text{max}-V_\text{min}) \sqrt{n}/\epsilon) \rceil +5$, where $\epsilon$ represents the desired error (follows from Theorem~\ref{theorem:ShapleyError}).
Let the number of qubits in the player register be $n$, one qubit per player excluding player $i$.
Note, in practice, there will often be a qubit to represent player $i$; however, we omit this qubit for simplicity.
Let the number of qubits in the utility register be one.

Consider the partition,
\begin{equation} \label{eq:P_l}
    P_\ell=\left(t_\ell(k)\right)_{k=0}^{2^\ell},
\end{equation}
of the interval $[0,1]$, leveraged for Riemann sums in Figure~\ref{fig:binomialAreaApprox}, where,
\begin{equation} \label{eq:t_l(k)}
    t_\ell(k):=\sin^2\left( \frac{k \pi}{2^{\ell + 1}} \right), \mbox{ with } k=0,1,\ldots,2^{\ell}.
\end{equation}
This partition is computationally simple to implement on a quantum computer.
Let us define $w_\ell(k)$ to be the width of the $k$th subinterval of $P_\ell$, with $k=0,1,\ldots,2^{\ell}-1$,
\begin{equation} \label{eq:w}
    w_\ell(k) := t_\ell(k+1)-t_\ell(k).
\end{equation}

Let us prepare the partition register to be,
\begin{equation}\label{eq:D}
    D_\ell\ket{0}_\texttt{Pt}^{\otimes \ell} =\sum\limits_{k=0}^{2^\ell-1} \sqrt{w_\ell(k)} \ket{k}_\texttt{Pt}.
\end{equation}
Note that $\sqrt{w_\ell(k)}$ is real and positive for all $k$, and that $\sum_{k=0}^{2^\ell-1} w_\ell(k)=1$.
Thus, $D_\ell$ can be implemented as a unitary.
Any approximation of Equation~\eqref{eq:D}, with a maximum error for $w_\ell(k)$ of $\mathcal{O}(4^{-n})$, is compatible with the error analysis in Appendix~\ref{appendix:errorAnalysis}.

\begin{definition} \label{def:C_D}
    We denote $C_D(\ell)$ as the complexity to implement $D_\ell$ such that it satisfies,
    \begin{equation*}
         \max_{k\in \{0,1\}^\ell} \abs{\bra{k} D_\ell \ket{0}^{\otimes \ell} - w_\ell(k)} \leq \mathcal{O}\left(4^{-\ell}\right),
    \end{equation*}
\end{definition}
Unfortunately, approximating a state to such a high accuracy is currently computationally intensive to implement.
It is still an area of active research \cite{rattew2022preparing}.
At the very least, $D_\ell$ can be implemented without error in $(23/24)2^\ell$ CNOT gates \cite{plesch2011quantum}.
For the remainder of this work, we use a placeholder for the complexity \hbox{$C_D(\ell) \leq (23/24)2^\ell < 32(V_\text{max}-V_\text{min}) \sqrt{n} / \epsilon$}.

The new state of the quantum system becomes,
\begin{align*}
    \ket{\psi_{1a}} &= (D_\ell\otimes I^{\otimes n} \otimes I) \ket{\psi_0}\\
    &= \sum\limits_{k=0}^{2^\ell-1} \sqrt{w_\ell(k)} \ket{k}_\texttt{Pt} \ket{0}_\texttt{Pl} \ket{0}_\texttt{Ut}.
\end{align*}

\begin{figure}
    \centering
    \includegraphics[width=\columnwidth]{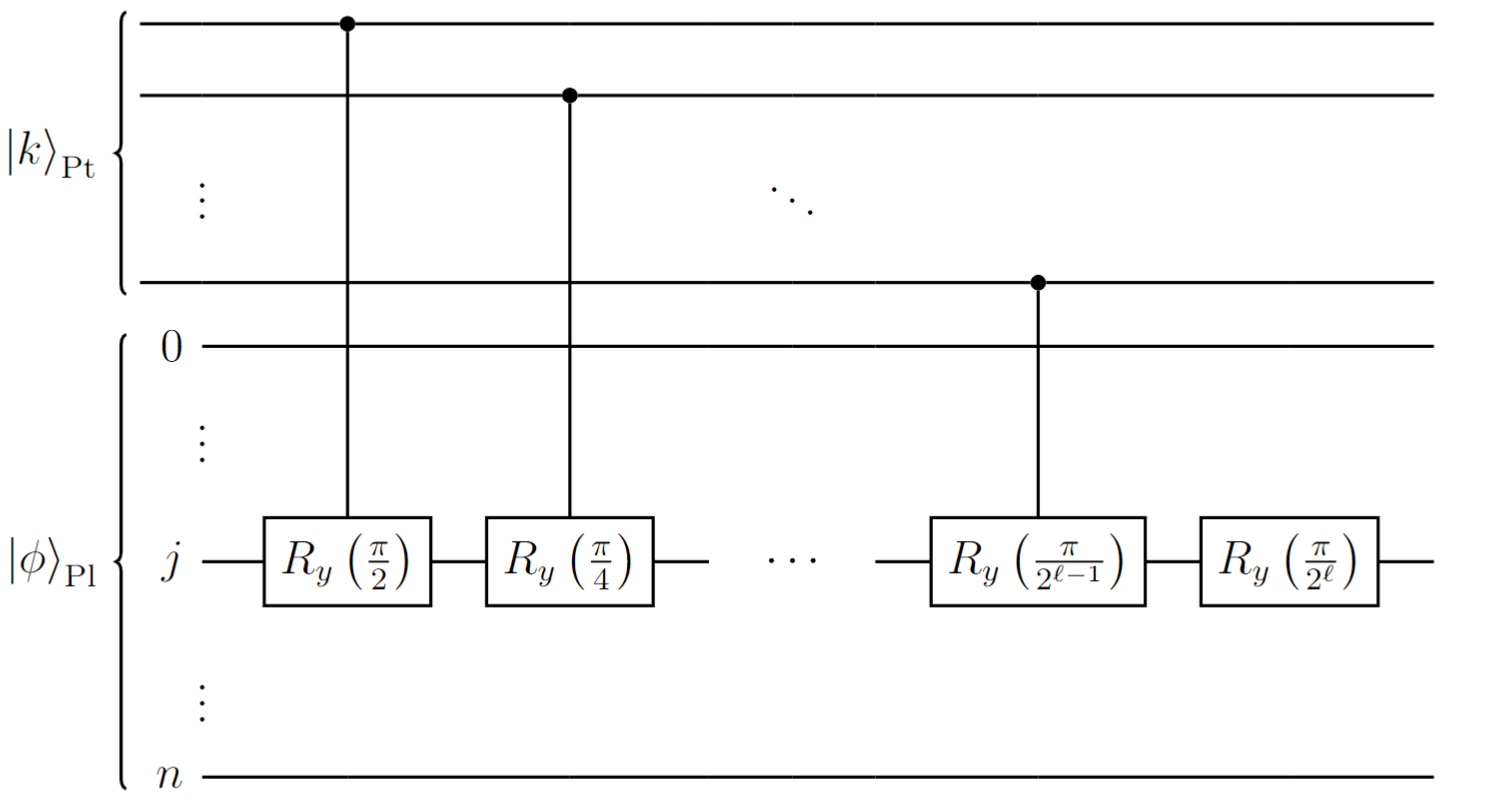}
    \caption{This circuit $R_j$ is a controlled rotation of the $j$th player qubit, where \hbox{$R_y(\theta) = (\cos(\theta/2), -\sin(\theta/2);\sin(\theta/2), \cos(\theta/2))$}.
    (Note: Library used for visualizing circuits can be found here in Ref.~\cite{kay_2023}).
    }
    \label{fig:controlledRotationCircuit}
\end{figure}

Define $R_j$ to be the $(\ell+n)$-qubit unitary controlled by the partition register acting as a rotation on the $j$th player qubit and acting as identity on every other player qubit (Figure~\ref{fig:controlledRotationCircuit}).
\begin{equation} \label{eq:R_j}
    R_j \ket{k}_\texttt{Pt}\ket{0}_\texttt{Pl}^{\otimes n}
    := \ket{k}_\texttt{Pt} \otimes \left(\ket{0}^{\otimes j} \otimes \left(\sqrt{1-t'_\ell(k)}\ket{0} + \sqrt{t'_\ell(k)}\ket{1} \right) \otimes \ket{0}^{\otimes n-j-1} \right)_\texttt{Pl}
\end{equation}
where $t'_\ell(k) = t_{\ell+1}(2k+1)$ is used to sample the height of the $k^\text{th}$ subinterval in the Riemann sum. 
Note that the value of $t'_\ell(k)$ is always somewhere in the range $[t_\ell(k),t_\ell(k+1)]$ (Appendix~\ref{appendix:tprimeProof}).

Unitary $R_j$ is performed on each qubit in the player register, controlled by the partition register.
The sum of all applications of $R_j$ can be performed with $n \ell$ gates in $\mathcal{O}(\max(n, \ell))$ layers and yields the quantum state:
\begin{align}\label{eq:psi2}
    \ket{\psi_{1b}} &= \prod\limits_{\substack{j=0}}^n (R_j \otimes I) \ket{\psi_{1a}} \\
    \label{eq:psi2_again}
    &= \sum\limits_{k=0}^{2^\ell-1} \sqrt{w_\ell(k)} \ket{k}_\texttt{Pt} \left( \sqrt{1-t_\ell'(k)}\ket{0} + \sqrt{t_\ell'(k)}\ket{1} \right)^{\otimes n}_\texttt{Pl} \ket{0}_\texttt{Ut}.
\end{align}
Recall that the player register is of size $n$ qubits.
By performing $R_j$ on each qubit, given a particular $\ket k_\texttt{Pt}$, each player qubit has been rotated by $t'_l(k)$.
The reason for doing this is non-trivial and justified in Lemma~\ref{lemma:hammingDistRearrange}.
In short, this operation is equivalent to sampling the heights of the $b_{n,m}$ function in every subinterval of the partition $P_\ell$. We use the resulting values in a Riemann sum approximation of the special beta function.

\begin{lemma}\label{lemma:hammingDistRearrange}
    Let $H_m$ be the set of binary numbers of hamming distance $m$ from $0$ in $n$ bits, such that $h\in H_m$ implies \hbox{$\abs{h}_H=m$} (Equation~\eqref{eq:hammingDist}).
    We have the following relation:
    \begin{equation*}
        \left(\sqrt{1-t'_\ell(k)}\ket{0} + \sqrt{t'_\ell(k)}\ket{1} \right)^{\otimes n} 
        = \sum\limits_{m=0}^n \sqrt{b_{n,m}\left(t'_\ell(k)        \right)} \sum\limits_{h\in H_m} \ket{h}.
    \end{equation*}
\end{lemma}
\begin{proof}
    We can rewrite our state as,
    \begin{equation*}
        \left(\sqrt{1-t'_\ell(k)}\ket{0} + \sqrt{t'_\ell(k)}\ket{1} \right)^{\otimes n} = \sum\limits_{h=0}^{2^n-1} \alpha_h \ket{h}.
    \end{equation*}
    where $\sum_{h=0}^{2^n-1} \abs{\alpha_h}^2 = 1$, $\alpha_h\in \mathbb{C}$.
    Noting that \hbox{$\cup_{m=0}^n H_m = \{0,1,\dots, 2^n-1\}$}, and $j\neq s$ implies $H_j \cap H_s = \emptyset$.
    We find,
    \begin{equation*}
        \sum\limits_{h=0}^{2^n-1} \alpha_h \ket{h} = \sum\limits_{m=0}^n \sum\limits_{h\in H_m} \alpha_h \ket{h}.
    \end{equation*}
    
    \noindent Taking $h=h_n\cdots h_{i+1} h_{i-1} \cdots h_0 \in \{0,1\}^{n}$, $\alpha_h$ can be defined as follows,
    \begin{align*}
        \alpha_h 
        &= \prod\limits_{\substack{j=0}}^{n} 
        \left( 
            \neg h_j \sqrt{1-t'_\ell(k)}
            + h_j \sqrt{t'_\ell(k)} 
        \right)\\
        &= \sqrt{
            t'_\ell(k)^{\abs{h}_H}
            \left( 1 - t'_\ell(k) \right)^{n-\abs{h}_H}
        }\\
        &= \sqrt{b_{n,\abs{h}_H}\big(t'_\ell(k) \big)}.
    \end{align*}
    where $\neg 0 = 1$ and $\neg 1 = 0$.
    The last step of the equality chain is based on Definition~\ref{def:betaFunction}.
    Since for all $h\in H_m$, $\abs{h}_H = m$, the result holds.
\end{proof}

Then, by Lemma~\ref{lemma:hammingDistRearrange}, we can rewrite $\ket{\psi_{1b}}$ as:
\begin{equation*}
    \ket{\psi_{1b}} = \sum\limits_{k=0}^{2^\ell-1} \sqrt{w_\ell(k)} \ket{k}_\texttt{Pt} \cdot \sum\limits_{m=0}^n \sqrt{b_{n,m}\left(t_\ell'(k) \right)} \cdot \sum\limits_{h\in H_m}\ket{h}_\texttt{Pl} \ket{0}_\texttt{Ut}
\end{equation*}
Note that, with this encoding style for $S_h$, $h$'s hamming distance from $0$ in $h$ bits equals the size of $S_h$.
In other words, if $h\in H_m$, then $S_h$ contains $m$ players. 

\begin{figure*}
    \centering
    \includegraphics[width=\textwidth]{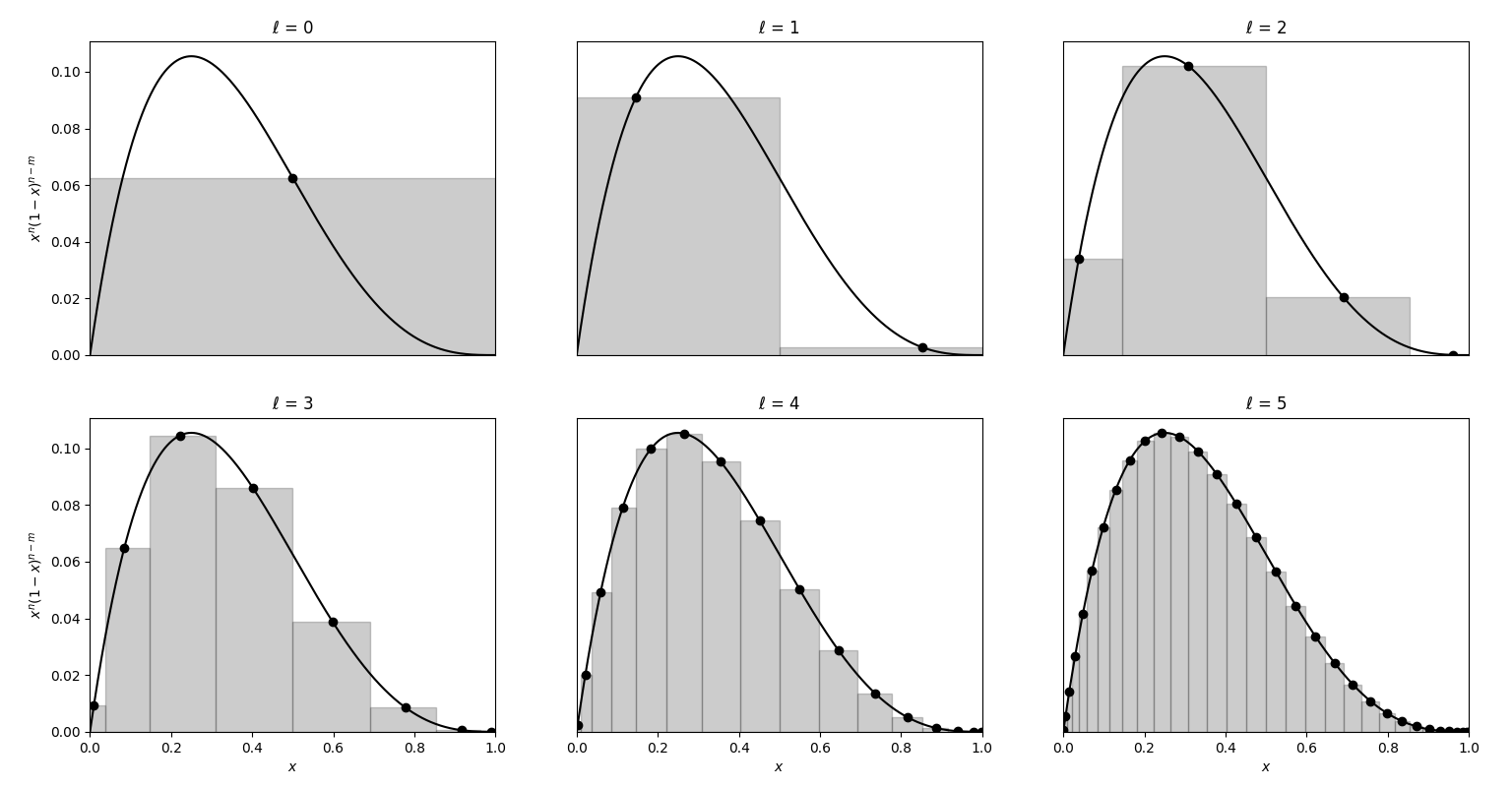}
    \caption{Visual representation of $\beta_{n,m}$ being approximated using Riemann sums of function $b_{n,m}(x)=x^m(1-x)^{n-m}$ over partition $P_\ell$, $t \in [0,1]$, $n=4$, $m=1$.
    The $k^\text{th}$ rectangle's height is $b_{n,m}(t'_\ell(k))=(t'_\ell(k))^m(1-~t'_\ell(k))^{n-m}$, and its width is $w_\ell(k)$. 
    }
    \label{fig:binomialAreaApprox}
\end{figure*}

Rearranging the quantum state $\ket{\psi_{1b}}$ gives,

\begin{equation*}
    \ket{\psi_{1b}} = \sum\limits_{m=0}^n \sum\limits_{h\in H_m} \sum\limits_{k=0}^{2^\ell-1} \sqrt{w_\ell(k) b_{n,m}\left(t_\ell'(k)\right)} \ket{k}_\texttt{Pt} \ket{h}_\texttt{Pl} \ket{0}_\texttt{Ut}.
\end{equation*}

Next, we perform $V^\pm$ on the utility register controlled by the player register.
Applying $U_V^\pm\ket{h}_\texttt{Pl}\ket{0}_\texttt{Ut}$ (Equation~\eqref{eq:U_V}) gives $\ket{\psi_2^\pm}$, which is equal to,

\begin{align*}
     \ket{\psi_2^\pm} &= \left(I^{\otimes \ell} \otimes U_V^\pm\right) \ket{\psi_{1b}} \\
     &= \sum\limits_{m=0}^n \sum\limits_{h\in H_m} \sum\limits_{k=0}^{2^\ell-1} \sqrt{w_\ell(k) b_{n,m}\left(t_\ell'(k)\right)} \ket{k}_\texttt{Pt} \ket{h}_\texttt{Pl} \ket{V^\pm(h)}_\texttt{Ut}.
\end{align*}

This operation is wholly dependent on the complexity of the game being analyzed.
Assuming the algorithm is implemented with a look-up table, one could likely use qRAM~\cite{giovannetti2008quantum}.
This approach has a time complexity of $\mathcal{O}(n)$ at the cost of $\mathcal{O}(2^n)$ qubits for storage.
However, depending on the problem, there are often far less resource-intense methods of implementing $U_V^\pm$, 
as is seen with the implementation of weighted voting games (Section~\ref{sec:example}).
In Appendix~\ref{appendix:errorAnalysis}, approximate implementations of $U^\pm_V$ and the resulting error propagation are analyzed.

This is the final quantum state.
Let us now analyze this state through the lens of density matrices.
Taking the partial trace with respect to the partition (\texttt{Pt}) and player (\texttt{Pl}) registers yields,

\begin{equation*}
    \tr_{\texttt{Pt},\texttt{Pl}}\left(\ketbra{\psi_2^\pm}{\psi_2^\pm}\right)=
    \sum\limits_{m=0}^n \sum\limits_{h\in H_m} \left(\sum\limits_{k=0}^{2^\ell-1} w_\ell(k) b_{n,m}\left(t_\ell'(k)\right) \right) \cdot \ket{V^\pm(h)}_\texttt{Ut} \bra{V^\pm(h)}_\texttt{Ut}.
\end{equation*}

\begin{theorem}
    The Riemann sum using partition $P_\ell$ to approximate area under $x^m(1-x)^{n-m}$ for $x\in[0,1]$ asymptotically approaches $\gamma(n,m)$.
    Formally,
    \begin{equation} \label{eq:riemannSumApproachesGamma}
        \sum\limits_{k=0}^{2^\ell-1} w_\ell(k) b_{n,m}\left(t_\ell'(k)\right) = \gamma(n,m) + \epsilon'
    \end{equation}
    where $\abs{\epsilon'} \leq \pi / 2^{\ell}$.
\end{theorem}
\begin{proof}
    Demonstrated in Appendix~\ref{appendix:errorAnalysis}, in particular, see Definition~\ref{def:gammaEll}, and Lemma~\ref{lemma:uUpperBound}.
\end{proof}

The Riemann sum approximation of the modified beta function, $\beta_{n,m}=\gamma(n,m)$ (Theorem~\ref{thm:beta=gamma}), is visualized in Figure~\ref{fig:binomialAreaApprox}.
Applying our approximation for $\gamma(n,m)$ and tracing out the $player$ and $partition$ registers, we have,
\begin{equation*}
    \tr_{\texttt{Pt},\texttt{Pl}}(\ketbra{\psi_2^\pm})\approx\sum\limits_{m=0}^n \sum\limits_{h\in H_m} \gamma(n,m) \ket{V^\pm(h)}_\texttt{Ut} \bra{V^\pm(h)}_\texttt{Ut}.
\end{equation*}

Finally, suppose we measure the utility register on a computational basis.
This yields the following expected value with error less than $(V_\text{max}-V_\text{min}) \sqrt{n} / 2^{\ell-3}$ (Appendix~\ref{appendix:errorAnalysis}), 
\begin{equation*}
    \sum\limits_{m=0}^n \sum\limits_{h\in H_m} \gamma(n,m) \hat{V}^\pm(h).
\end{equation*}

Comparing the expected values of measuring the utility registers for $\ket{\psi_2^+}$ and $\ket{\psi_2^-}$ gives,
\begin{equation*}
    \sum\limits_{m=0}^n \sum\limits_{h\in H_m} \gamma(n,m) \hat{V}^+(h) - \sum\limits_{m=0}^n \sum\limits_{h\in H_m} \gamma(n,m) \hat{V}^-(h).
\end{equation*}
Plugging in the definition for $\hat{V}^\pm$, we have,
\begin{equation*}
    \frac{1}{V_\text{max} - V_\text{min}} \sum\limits_{m=0}^n \sum\limits_{h\in H_m} \gamma(n,m) \left(V\left(S_h \cup \{i\}\right) - V\left(S_h\right)\right).
\end{equation*}

Notice that in the $S_h$ encoding, $H_m$ represents each subset of $F\setminus\{i\}$ of size $m$.
As a result, the equation is in effect, summing over each subset of $F\setminus\{i\}$ with corresponding $\gamma$ weights, Definition~\ref{def:shapleyvalue}.
Thus, multiplying the previous equation by $(V_\text{max}-V_\text{min})$, we obtain,
\begin{equation*}
    \sum\limits_{S\subseteq F\setminus\{i\}} \gamma(\lvert F \setminus \{i\} \rvert, \lvert S \rvert) \cdot \left(V\left(S \cup \{i\}\right) - V\left(S\right)\right).
\end{equation*}

This expected value is precisely the Shapley value $\Phi(i)$ to some error analytically upper-bounded in Appendix~\ref{appendix:errorAnalysis}.
With the ability to craft these states, we can now extract the required information to find a close approximation to the Shapley value.
Assuming we could get the expected value instantly, we would have an error of less than $16(V_\text{max}-V_\text{min})\sqrt{n}/2^\ell$.
However, it takes some work to approximate an expected value.
This can be achieved with accuracy $\epsilon$ with some chosen confidence using amplitude estimation as described by Montanaro~\cite{montanaro2015quantum}, resulting in a circuit which repeats Steps $1$ and $2$ in the order of $\mathcal{O}(\epsilon^{-1})$ times.

To conclude, by Appendix~\ref{appendix:errorAnalysis} Theorem~\ref{theorem:totalComplexity}, we have time complexity in $\mathcal{O} \left( \lambda (C_D(\lceil\log_2 (\lambda\sqrt{n}) \rceil) + C_V(\lambda^{-1}) + n\log_2 (\lambda\sqrt{n}) ) \right)$. 
Where $\lambda$ is equal to $(V_\text{max}-V_\text{min})/\epsilon$, $\epsilon$ is the desired maximum error with fixed likelihood of success, $C_D(\lceil\log_2 (\lambda\sqrt{n}) \rceil)$ is the complexity of implementing $D_{\lceil\log_2 (\lambda\sqrt{n}) \rceil}$ with max error $\mathcal{O}(\lambda^{-2} n^{-1})$ (Definition~\ref{def:C_D}), and $C_V(\lambda^{-1})$ is the complexity of implementing $U^\pm_V$ with maximum error of $\lambda^{-1}$ (Appendix~\ref{appendix:errorAnalysis}, Definition~\ref{def:C_V}).
Assuming a case where $C_V(\lambda^{-1})$ is proportional to or larger than $C_D(\log_2 (\lambda\sqrt{n}))+n\log_2 (\lambda\sqrt{n})$, 
the difference $V_\text{max} - V_\text{min}$ is proportional to $\sigma$ and $C_V(\epsilon/(V_\text{max}-V_\text{min}))$ is proportional to $C_V(\epsilon)$, we have a quadratic improvement over Monte-Carlo methods.
In the following section, we demonstrate concrete examples of the algorithm.

\section{Quantum Shapley Value Examples}
\label{sec:example}

\subsection{Weighted Voting Games}
\noindent Perhaps we can help David solve his problem (cf. Subsection~\ref{sec:WeightedVotingGame}) using our quantum approach. 
We intend to apply the method presented in Section~\ref{sec:algorithm} for weighted voting games. Additional results, together with the simulation code, are available in a \href{https://github.com/iain-burge/QuantumShapleyValueAlgorithm}{companion GitHub repository} \cite{githubEntry}.
Let us approximate each player's Shapley value, $\Phi(i)$.
We have a game $G=(F,V)$, where $F=\{0,1,2\}$, $n=2$, and $V$ is defined in Equation~\eqref{eq:voteGameV}.
Let the voting weights be $w_0=3$, $w_1=2$, and $w_2=1$.
Thus, we can define $V^\pm(h)$, $h\in \{0,1\}^2$, where $h$ represents an element in $\mathcal{P}(F\setminus \{i\})$,

\begin{align*}
    V^-(h)&= 
    \begin{cases}
        1 &\quad \text{if } \sum_{s\in S_h} w_s \geq q\\
        0 &\quad \text{otherwise}
    \end{cases}
    & V^+(h)&= 
    \begin{cases}
        1 &\quad \text{if } \sum_{s\in S_h \cup \set{i}} w_s \geq q\\
        0 &\quad \text{otherwise}
    \end{cases}
\end{align*}
Note that, $V^\pm(h)$ is either $0$ or $1$.
Thus, we can take $V_{\max}=1,V_{\min}=0$.
We define $U_V^\pm$ to be:
\begin{equation*}
    U_V^\pm \ket{x}\ket{0} = \ket{x} \otimes \left[\left(1-\hat{V}^\pm(x)\right)\ket{0} + \hat{V}^\pm(x)\ket{1}\right]
\end{equation*}
The quantum circuit for the scenario is shown in Figure~\ref{fig:U_wOfVotingGame}.
Note that, much like with the figure, it is trivial to generate $V^+$ and $V^-$ if one has a quantum gate to compute $V$.

\begin{figure}
    \centering
    \includegraphics[width=\columnwidth]{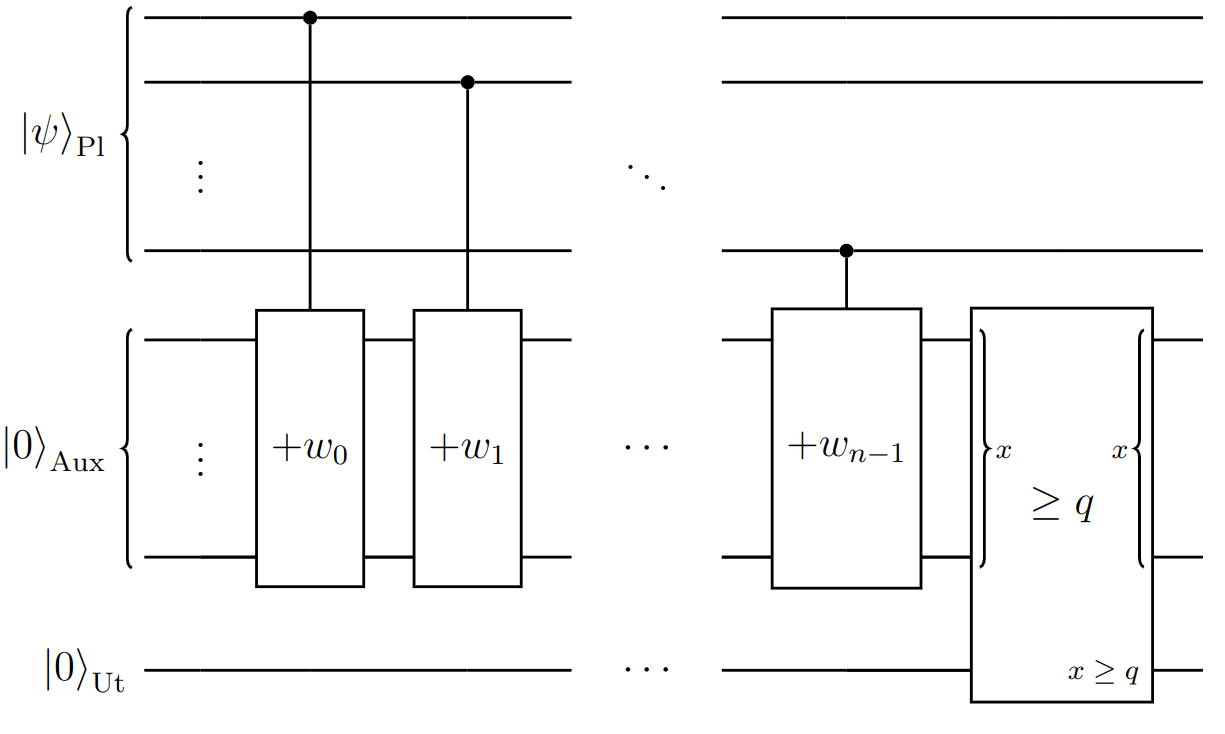}
    \caption{Circuit of $U_V^\pm$ for a weighted voting game.
    This circuit takes an basis state input $\ket {h_n, \dots, h_{i+1}, x, h_{i-1}, \dots, h_0}$ and outputs $\ket{V^-(S_h)}$ when $x=0$ or $\ket{V^+(S_h)}$ when $x=1$ to the utility register (Recall, $S_h$ is defined in Definition~\ref{def:S_h}).
    The auxiliary register contains the total vote count.
    Just before the $\geq q$ gate, the $Aux$ register is in a basis state corresponding to the vote count of $S_h$, including or excluding player $i$s vote depending on $x$.
    The $\geq q$ gate uses the auxiliary register as an input and outputs whether the vote count exceeds threshold $q$ in the \texttt{Ut} register. 
    After this gate, it is trivial to clear the $Aux$ register by subtracting each player's contribution.
    Results and simulation code are available in a \href{https://github.com/iain-burge/QuantumShapleyValueAlgorithm}{companion GitHub repository} \cite{githubEntry}. 
    }
    \label{fig:U_wOfVotingGame}
\end{figure}

With $U_V^\pm$ defined, all other steps are entirely agnostic to voting games.
Let $\ell$ be equal to $2$, and suppose we can instantly extract the expected value for simplicity's sake.
This is not a realistic scenario, but it demonstrates how quickly the expected value of the utility register converges in Steps~1 and~2.
With these parameters, we get the following approximations for the Shapley values:
\begin{equation*}
    \tilde\Phi_0 \approx 0.6617,\quad \tilde\Phi_1, \tilde\Phi_2 \approx 0.1616 
\end{equation*}
The direct calculation for Alice can be seen in Appendix~\ref{appendix:byHandExample}. 

To rigorously demonstrate efficacy, we performed many trials on random weighted voting games (Figure~\ref{fig:expectedValueError}).
Visual inspection confirms that the error from Step~1 decreases exponentially with respect to $\ell$.
Step 2 depends entirely on the game. However, if it is possible to implement on a classical computer, much like with Grover's algorithm, it can be implemented in a quantum setting~\cite{grover1996fast}.
Step 3 is well studied, it extracts the expected value of measuring the utility register with error $\epsilon$, with a fixed probability of success.
Use amplitude estimation as described in Section~\ref{sec:algorithm} with the techniques described in \cite{montanaro2015quantum}. 

\begin{figure*}
    \centering
    \includegraphics[width=\textwidth]{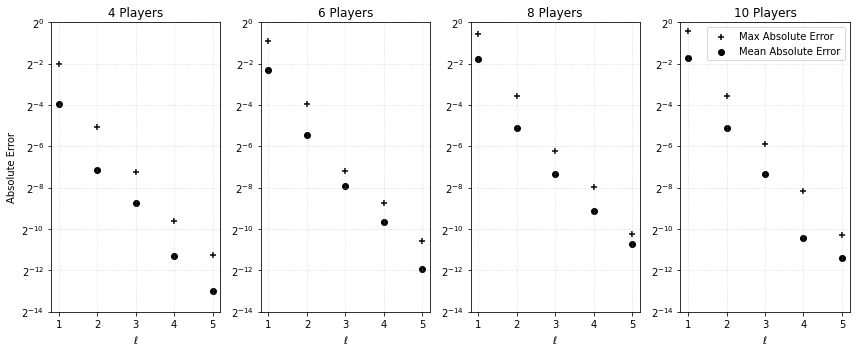}
    \caption{
        This figure demonstrates the exponentially small error introduced in Step~1 with respect to $\ell$.
        We generated $64$ random weighted voting games for each condition, e.g., for each combination of $\ell$ and number of players.
        Random weights $w_j\in \mathbb{N}$ were assigned for each case such that $q\leq \sum_j w_j < 2q$.
        There were four primary scenarios: 
        (1) Four players, voting threshold $q=8$; 
        (2) Six players, voting threshold $q=16$;
        (3) Eight players, voting threshold $q=32$; and
        (4) Ten players, voting threshold $q=32$.
        We approximated every player's Shapley value for each scenario with our quantum algorithm using various $\ell$'s, taking only Step~1's error into account.
        Next, we found the absolute error of our approximation by comparing each approximated Shapley value to its true value.
        The graphs show the mean and maximum absolute errors of each condition.
        Results and simulation code are available in a \href{https://github.com/iain-burge/QuantumShapleyValueAlgorithm}{companion GitHub repository} \cite{githubEntry}.
    }
    \label{fig:expectedValueError}
    
\end{figure*}

%======================================================%
%======================================================%
\subsection{Explainability and Shapley Values of Quantum Binary Classifiers}
\label{section:explainableQuantumAI}
\noindent The local explanation method described in Subsection~\ref{sec:XAI} is conveniently translated to the quantum context.
Suppose we are given a quantum classifier $U_C$ that acts on two registers as follows, $U_C\ket{j}_\texttt{Pl}\ket{0}_\texttt{Ut}\rightarrow \ket{j}_\texttt{Pl}\ket{C(j)}$, where $j\in \{0,1\}^{r\times r}$, and $C$ is defined in Subsection~\ref{sec:XAI}.
To calculate the $i$th player's Shapley value $\Phi_i$ of $V_{C,x}$, Equation~\eqref{eq:localExplainability}, and hence generate a local explanation for the model decision $C(x)=y$, we only need a simple modification to our method.
We now give a sketch of the implementation.
Suppose that, after Step~1 of the Section~\ref{sec:algorithm} algorithm, we have the approximate state,
\begin{equation*}
    \tr_{\texttt{Pt}}\ketbra{\psi_{1b}} \approx \sum\limits_{m=0}^n \sum\limits_{h\in H_m} \gamma(n,m) \ket{h}_\texttt{Pl}\bra{h}_\texttt{Pl},
\end{equation*}
ignoring the \texttt{Ut} register.
To craft the distribution which gives value function $V_{x,C}$, Equation~\eqref{eq:localExplainability}, we add a second player register $\ket{0}^{\otimes n}_\texttt{Pl'}$.
We can then apply \texttt{NOT} gates to each qubit in \texttt{Pl}.
Next, apply controlled Hadamard gates between the previous and new player registers.
In particular for each player $j$, the $j$th qubit in the \texttt{Pl} register controls a Hadamard transform on the $j$th qubit of the \texttt{Pl'} register.
We proceed with old player register $\texttt{Pl}$ traced out. 
In the remaining Steps, we treat \texttt{Pl'} as the player register.
If we continue with the other steps of the algorithm, Steps~2 and~3, the output will be the Shapley value $\Phi_i$ of corresponding to $V_{C,x}$.

\section{Improved Quantum Algorithm for Shapley Value Approximation}
\label{section:qShapImproved}

\medskip

\noindent Section~\ref{sec:algorithm} describes an algorithm to estimate the Shapley values of cooperative games.
Assuming a reasonable implementation of $D_\ell$~\eqref{eq:D}, the algorithm is often more efficient than classical random sampling.
However, an implementation of $D_\ell$ may not be realistic in the short term, as it might require repeatedly calculating the square root of $\sin$ \cite{rattew2022preparing}.
Additionally, current techniques may not converge quickly enough to $D_\ell$ to give an advantage over classical methods.
As a result, we want to avoid the use of $D_\ell$.

Previously, the matrix $D_\ell$ was used to compensate for the peculiar partition $P_\ell$ \eqref{eq:P_l}, a consequence of using the circuit from Figure~\ref{fig:controlledRotationCircuit}.
If we can create a more reasonable partition, we will not need such expensive compensations.
In particular, we can leverage the Quantum CORDIC algorithm for approximation of \texttt{arcsin} to sample from $b_{n,m}(x)$ uniformly \cite{burge2024quantum}.
This can be combined with the uniform partition,
\begin{equation}\label{eq:partitionQ_ell}
    Q_\ell=\left( k2^{-\ell} \right)_{k=0}^{2^\ell},
\end{equation}
to approximate the $\gamma$ Shapley weights, Definition~\ref{def:shapleyvalue}, 
The partition $Q_\ell$ has the desirable property that each subinterval has a width equal to $2^{-\ell}$, which is easily implemented with Hadamard gates.

Consider the $n+1$ player game as defined in Section~\ref{sec:algorithm}, we wish to find the Shapley value $\Phi_i$ of the $i$th player.
By \cite{burge2024quantum}, leveraging $4\ell+\mathcal{O}(\log\ell)$ zeroed auxiliary qubits, and using $\mathcal{O}(\ell^2)$ \texttt{CNOT}s, we can perform the following operation,
\begin{equation*}
    \frac{1}{\sqrt{2^\ell}}\sum\limits_{k=0}^{2^\ell-1} \ket{k}_\texttt{Pt} \ket{0} \rightarrow 
    \frac{1}{\sqrt{2^\ell}}\sum\limits_{k=0}^{2^\ell-1} \ket{k}_\texttt{Pt} \left(\sqrt{1-r_\ell(k)}\ket{0} + \sqrt{r_\ell(k)}\ket{1} \right),
\end{equation*}
where $r_\ell(k)=(k+1/2)2^{-\ell}+\epsilon_k$, and $\abs{\epsilon_k}$ is less than $2^{-(\ell+1)}$.
It follows that $r_\ell(k)\in\left[k2^{-\ell},(k+1)2^{-\ell}\right]$.
This can be applied to each player qubit, yielding state,
\begin{equation*}
    \frac{1}{\sqrt{2^\ell}}\sum\limits_{k=0}^{2^\ell-1} \ket{k}_\texttt{Pt} \left(\sqrt{1-r_\ell(k)}\ket{0} + \sqrt{r_\ell(k)}\ket{1} \right)_\texttt{Pl}^{\otimes n}.
\end{equation*}

\noindent Using a nearly identical argument to Lemma~\ref{lemma:hammingDistRearrange}, and some rearrangement, this is equal to,
\begin{equation*}
    \sum\limits_{m=0}^n \sum\limits_{H_m} 
    \sum\limits_{k=0}^{2^\ell-1} \sqrt{\frac{b_{n,m}\left(r_\ell(k)\right)}{2^\ell}}
    \ket{k}_\texttt{Pt} \ket{h}_\texttt{Pl},
\end{equation*}
where $b_{n,m}$ is defined in Definition~\ref{def:betaFunction}, and $H_m$ is defined in the statement of Lemma~\ref{lemma:hammingDistRearrange}.
We next add the \texttt{Ut} register and perform $U_V^\pm$, Equation~\eqref{eq:U_V}, giving state,
\begin{equation*}
    \ket{\phi^\pm} = \sum\limits_{m=0}^n \sum\limits_{H_m} 
    \sum\limits_{k=0}^{2^\ell-1} \sqrt{\frac{b_{n,m}\left(r_\ell(k)\right)}{2^\ell}}
    \ket{k}_\texttt{Pt} \ket{h}_\texttt{Pl} \ket{V^\pm(h)}_\texttt{Ut}.
\end{equation*}
Similarly to Section~\ref{sec:algorithm}, this yields the following expected value for measuring the utility register,
\begin{equation*}
    \sum\limits_{m=0}^n \sum\limits_{h\in H_m} \left(\sum\limits_{k=0}^{2^\ell-1} \frac{b_{n,m}\left(r_\ell(k)\right)}{2^\ell}\right) \hat{V}^\pm(h).
\end{equation*}

Notice that $\sum_{k=0}^{2^\ell-1} 2^{-\ell}b_{n,m}\left(r_\ell(k)\right)$ is a Riemann sum (Definition~\ref{def:riemannSum}) approximation of $\beta_{n,m}$ (Definition~\ref{def:betaFunction}), which is equal to $\gamma(n,m)$ (Theorem~\ref{thm:beta=gamma}).
Combining the reasoning of Corollary~\ref{corollary:ErrorBound} and Lemma~\ref{lemma:uUpperBound}, our Riemann sum approximates $\gamma(n,m)$ with absolute error less than $2^{1-\ell}b_{n,m}(m/n)$.
As a result, our final error using this method will be less than the bound for the previous algorithm described in Section~\ref{sec:algorithm}, thus, the error analysis of Appendix~\ref{appendix:errorAnalysis} still applies for any given $\ell$.

\begin{theorem}\label{theorem:improvedShapComplexity}
    The expected value of the \texttt{Ut} register of $\ket{\phi^\pm}$, multiplying by $V_{\max}-V_{\min}$, and subtracting by $V_{\min}$ yields,
    \begin{equation*}
        \Phi_i + \rho,\quad \rho\in[-\epsilon,\epsilon],\quad \epsilon>0.
    \end{equation*}
    Given a fixed probability of success, the operation takes,
    \begin{equation*}
        \mathcal{O}\left( 
        \left[\frac{\sqrt{\Delta V(\Phi_i-V_{\min})}}{\epsilon}\right] \cdot
        \left[ \log^2\left( \frac{\Delta V n}{\epsilon} \right) 
        + n \log\left( \frac{\Delta V n}{\epsilon} \right) 
        + C_V\left( \frac{\epsilon}{4\Delta V} \right) \right]
        \right),
    \end{equation*}
    operations.
    Wherein, $\Delta V=V_{\max}-V_{\min}$, and $C_V(x)$ is the complexity of implementing $U_V^\pm$ with maximum error bound $x$ (Definition~\ref{def:C_V}).
\end{theorem}
\begin{proof}
    Follows from Theorem~\ref{theorem:totalComplexity} and the replacement of transformation $D_\ell$ with the transformation described in \cite{burge2024quantum}.
    In particular, we must perform $\mathcal{O}\left(\epsilon^{-1}\sqrt{\Delta V(\Phi_i-V_{\min})}\right)$ iterations of amplitude estimation as described by Montanaro \cite{montanaro2015quantum}.
    To prepare the $\gamma$ distribution, each iteration requires that we apply \texttt{arcsin} as described in \cite{burge2024quantum}, requiring $\mathcal{O}(\log^2(\Delta Vn/\epsilon))$ steps.
    The results are encoded in player amplitudes using $\mathcal{O}(n\log(\Delta Vn/\epsilon))$ operations.
    Finally, each iteration requires a query to the value function $U_V^\pm$ with maximum error $\epsilon/(4\Delta V)$, taking $\mathcal{O}(C_V(\epsilon/(4 \Delta V)))$ operations (Definition~\ref{def:C_V}).
\end{proof}

\begin{table}[] 
\centering
\caption{Complexity Comparison \label{tabel:complexities}}
\begin{tabular}{l l l} 
 \hline
  Quantum Algorithm & : & $\mathcal{O}\left(
        \frac{\sqrt{(V_{\max}-V_{\min})(\Phi_i-V_{\min})}}{\epsilon}
        C_V\left( \frac{\epsilon}{4(V_{\max}-V_{\min})} \right)
        \right) $\\
  Monte-Carlo & : & $\mathcal{O}\left(\frac{\sigma^2}{\epsilon^2} C_V(\epsilon) \right)$ \\
 \hline
\end{tabular}
\end{table}

When the range is not exponentially large, and $U_V$ is even somewhat complex with respect to number of players and error, the $\gamma$ preparation steps can be ignored in the complexity analysis.
\begin{corollary}\
    Suppose $C_V(\epsilon/(4\Delta V))$ is greater or equal to $\mathcal{O}(\log^2(\Delta Vn/\epsilon)+n\log(\Delta Vn/\epsilon))$, $\Delta V=V_{\max}-V_{\min}$. 
    The expected value of the \texttt{Ut} register of $\ket{\phi^\pm}$, multiplying by $V_{\max}-V_{\min}$, and subtracting by $V_{\min}$ yields,
    \begin{equation*}
        \Phi_i + \rho,\quad \rho\in[-\epsilon,\epsilon],\quad \epsilon>0.
    \end{equation*}
    Given a fixed probability of success, the operation takes,
    \begin{equation*}
        \mathcal{O}\left(
        \frac{\sqrt{\Delta V(\Phi_i-V_{\min})}}{\epsilon}
        C_V\left( \frac{\epsilon}{4\Delta V} \right)
        \right),
    \end{equation*}
    operations.
    $C_V(x)$ is defined in Definition~\ref{def:C_V}.
    This is compared with Monte Carlo methods in Table~\ref{tabel:complexities}.
\end{corollary}

\begin{example}
    Let us give a more concrete example, suppose we have a game $G=(F,V)$, where $F$ is an $n+1$ player game and we have a value function $V:\mathcal{P}(F)\rightarrow \{0,1\}$.
    Then $V_{\min},V_{\max}$ are $0$ and $1$ respectively.
    Suppose that $U_V^\pm$ is, as a result, implemented with no error.
    In this case, the complexity $C_V=C_V(0)$ is fixed (Definition~\ref{def:C_V}), and we assume that $C_V$ is proportional to the complexity of implementing $V$ classically.
    Suppose we wish to find the $i$th player's Shapley value $\Phi_i$.
    We proceed assuming $\Phi_i$ is on the interval $[-1/2,1/2]$.
    Thus, using classical Monte Carlo methods, $V$ yields a Bernoulli distribution with a variance of $\sigma^2 =\Phi_i(1-\Phi_i)$, since $\abs{\Phi_i}\leq 0.5$, it follows that $\sigma^2\in\mathcal{O}(\Phi_i)$.
    We therefore have classical complexity,
    \begin{equation*}
        \mathcal{O}\left( \frac{\Phi_i}{\epsilon^2}  C_V\right).
    \end{equation*}
    On the other hand, in the quantum context, we have,
    \begin{equation*}
        \mathcal{O}\left( \frac{\sqrt{\Phi_i}}{\epsilon} \left(C_V + \log^2\left(\frac{1}{\epsilon} \right) \right) \right).
    \end{equation*}
    This follows directly from Theorem~\ref{theorem:improvedShapComplexity}.
    Note that $n$ is fixed in this context.
    This is a quadratic improvement up to polylogarithmic factors.
\end{example}

\section{Related Work}
\label{section:relwork}

\medskip

\noindent Explainability in \gls*{ai} is a quickly growing area of importance, especially given the rise in legislative interest and regulation for transparency in the United States and European Union respectively \cite{nannini2023explainability, goodman2017european, nisevic2024explainable}.
The post-hoc method of additive explanations in model agnostic contexts is addressed in Riberio et al.'s seminal work \cite{ribeiro2016should}.
Riberio et al. introduced the model agnostic \gls*{lime} algorithm to assess the importance of each feature for particular decisions.
It leveraged the fact that many classifiers, represent differential functions, and are hence locally linear.
In addition to Riberio et al.'s work, various additive explanation methods have been introduced, such as layer-wise relevance propagation \cite{bach2015pixel}. 
Lundberg and Lee \cite{lundberg2017unified} found that many of the newly introduced additive explanation methods were equivalent to finding the Shapley values of the input variables.
Specifically, each explanation which satisfies a small set of properties is equivalent to a Shapley value approximation.
The work culminated in Lundberg and Lee's widely used \gls*{shap} algorithm.
More recently, the particular implementation of \gls*{shap} has come under scrutiny.
According to Marques-Silva and Huang \cite{marques2024explainability}, there are no rigorous error guarantees for the most popular approximations of \gls*{shap}.
More concerning, Marques-Silva and Huang claim \gls*{shap}'s definition is not formally justified and demonstrates specific circumstances where exact \gls*{shap} produces undesirable results.
Their criticisms do not necessarily translate to Shapley values but illustrate the need for deep care and caution for robust explainability.

In tandem with additive explanations, there are multiple approaches which aim to make \gls*{ai} more trustworthy.
In the domain of post-hoc explanations, there is also the option of counterfactual explanations \cite{byrne2019counterfactuals}.
A counterfactual expresses how an outcome would be different given different initial conditions.
A counterfactual explanation describes what minimal changes could be made to get a more desirable outcome, for example, ``\emph{your bank loan would have been accepted if your income had been $10\%$ higher}".
On the other hand, some authors argue that, in high-stakes situations, black box models should be abandoned altogether in favour of inherently interpretable models \cite{rudin2019stop, ghassemi2021false}.
These models include purely linear models, decision trees and more.
While some believe black-box models have a performance advantage in many domains, Rudin \cite{rudin2019stop} claims that interpretable models have similar potential in many relevant circumstances. 
Rudin argues that interpretable models allow for better iterative improvements, as their flaws are understandable and therefore easier to address.
Additionally, when the data is structured, and effectively preprocessed, Rudin asserts that simple classifiers perform similarly to complex classifiers.
However, London~\cite{london2019artificial} argues that in cases where black-box models have higher accuracy, they should not be automatically discarded.

Beyond explainable \gls*{ai}, Shapley values have been widely used to address multiple engineering problems, including regression, statistical analysis, and \gls*{ml}~\cite{lipovetsky2023quantum}.
Finding Shapley values presents a difficult computational combinatorial problem. 
The deterministic computation of Shapley values in weighted voting games is at least as difficult as NP-Hard~\cite{matsui2001np, prasad1990np}. 
Since voting games are some of the simplest cooperative games, this result does not bode well for more complex scenarios. 
In the context of Shapley values for machine learning, it has also been shown that the calculation of Shapley values is not tractable for regression models \cite{van2022tractability}. 
Similarly, on the empirical distribution, finding a Shapley value takes exponential time~\cite{bertossi2020causality}.
While the general case direct Shapley value calculation is extremely computationally complex, this is not the case if we weaken our requirements.
Calculating the Shapley for specific games allows for optimizations which leverage the structure of the problem.
This has been done, for instance, in the context of games on graphs \cite{tarkowski2017game}.
The other approach, which is relevant to our work, is based on probabilistic approximations.
Specifically, Castro et al. \cite{castro2009polynomial} describes a method based on Monte Carlo methods to approximate Shapley values in polynomial time.
By Chebyshev's inequality, Monte Carlo methods have a query complexity of $\mathcal{O}(\sigma^2/\epsilon^2)$, which is tenable when compared to naive and exponentially complex methods \cite{montanaro2015quantum}.

Simultaneously, there has been a strong research effort into quantum AI. 
Quantum principal component analysis is an early example of an efficient quantum algorithm in the domain of AI \cite{lloyd2014quantum}. 
This work lead Rebentrost et al.'s \cite{rebentrost2014quantum} algorithm for quantum support vector machines, which Rebentrost et al. claim is exponentially faster than known classical methods in some cases.
It has been shown that training quantum neural networks is possible and that they are universal function approximations \cite{beer2020training}.
In response to the progress in quantum \gls*{ai}, there has been steady growth in the topic of Quantum Explainable AI.
Treating the quantum algorithm as a black box allows for the use of \gls*{lime} \cite{deshmukh2023explainable, pira2024interpretability}.
In a similar vein, \gls*{shap} and Shapley values have both been leveraged to explain the behaviour of quantum circuits \cite{steinmuller2022explainable, heese2023explaining}, though these works do not leverage quantum effects to accelerate the computation.

\section{Conclusion}
\label{sec:conclusion}

\noindent We have introduced a quantum algorithm which allows for a more efficient approximation of Shapley values.
The algorithm is often quadratically faster than is possible with classical Monte Carlo methods, up to polylogarithmic factors.
In other words, doubling work yields double the precision. 
In contrast to the quadrupling of work required classically.
The algorithm has two potential applications.
First, it can be applied directly to accelerate the calculation of Shapley values in cooperative games.
Second, we can leverage the algorithm to construct additive explanations of quantum circuits.

For future work, examining practical examples of the quantum algorithm for Shapley values is a source of various interesting problems.
This is especially true in cases where the value function can be more efficiently computed on a quantum computer.
In a similar vein, specific games can be analyzed so that their structure might be leveraged for larger quantum advantages.
For quantum explainability, there are several useful directions to explore.
An additional area of inquiry is whether certain models, such as quantum support vector machines, can have their Shapley values computed more efficiently.
Finally, our algorithm can be used to better understand proposed quantum AIs.

\bibliographystyle{unsrt_MS}
\bibliography{main}

\begin{thebibliography}{10}

\bibitem{goodman2017european}
B.~Goodman and S.~Flaxman.
\newblock {European Union regulations on algorithmic decision-making and a
  ``right to explanation''}.
\newblock {\em AI magazine}, 38(3):50--57, 2017.

\bibitem{nisevic2024explainable}
M.~Nisevic, A.~Cuypers, and J.~De~Bruyne.
\newblock Explainable ai: Can the ai act and the gdpr go out for a date?
\newblock In {\em 2024 International Joint Conference on Neural Networks
  (IJCNN)}, pages 1--8. IEEE, 2024.

\bibitem{nannini2023explainability}
L.~Nannini, A.~Balayn, and A.~L. Smith.
\newblock Explainability in ai policies: A critical review of communications,
  reports, regulations, and standards in the eu, us, and uk.
\newblock In {\em Proceedings of the 2023 ACM conference on fairness,
  accountability, and transparency}, pages 1198--1212, 2023.

\bibitem{rudin2019stop}
C.~Rudin.
\newblock {Stop explaining black box machine learning models for high stakes
  decisions and use interpretable models instead}.
\newblock {\em Nature Machine Intelligence}, 1(5):206--215, 2019.

\bibitem{burgeQCE2023}
I.~Burge, M.~Barbeau, and J.~Garcia-Alfaro.
\newblock Quantum algorithms for shapley value calculation.
\newblock In {\em 2023 IEEE International Conference on Quantum Computing and
  Engineering (QCE)}, volume~1, pages 1--9. IEEE, 2023.

\bibitem{lundberg2017unified}
S.~M. Lundberg and S.-I. Lee.
\newblock A unified approach to interpreting model predictions.
\newblock {\em Advances in neural information processing systems}, 30, 2017.

\bibitem{matsui2001np}
Y.~Matsui and T.~Matsui.
\newblock {NP-completeness for calculating power indices of weighted majority
  games}.
\newblock {\em Theoretical Computer Science}, 263(1-2):305--310, 2001.

\bibitem{prasad1990np}
K.~Prasad and J.~S. Kelly.
\newblock {NP-completeness of some problems concerning voting games}.
\newblock {\em International Journal of Game Theory}, 19(1):1--9, 1990.

\bibitem{castro2009polynomial}
J.~Castro, D.~G{\'o}mez, and J.~Tejada.
\newblock {Polynomial calculation of the Shapley value based on sampling}.
\newblock {\em Computers \& Operations Research}, 36(5):1726--1730, 2009.

\bibitem{biamonte2017quantum}
J.~Biamonte, P.~Wittek, N.~Pancotti, P.~Rebentrost, N.~Wiebe, and S.~Lloyd.
\newblock Quantum machine learning.
\newblock {\em Nature}, 549(7671):195--202, 2017.

\bibitem{heese2023explaining}
R.~Heese, T.~Gerlach, S.~Mücke, S.~Müller, M.~Jakobs, and N.~Piatkowski.
\newblock Explaining quantum circuits with shapley values: Towards explainable
  quantum machine learning, arXiv: 2301.09138,
  \url{https://doi.org/10.48550/arXiv.2301.09138}, March 2023.

\bibitem{burge2023quantum}
I.~Burge, M.~Barbeau, and J.~Garcia-Alfaro.
\newblock {A Quantum Algorithm for Shapley Value Estimation}, arXiv:2301.04727,
  \url{ https://doi.org/10.48550/arXiv.2301.04727}, March 2023.

\bibitem{deshmukh2023explainable}
S.~Deshmukh, B.~K. Behera, P.~Mulay, E.~A. Ahmed, S.~Al-Kuwari, P.~Tiwari, and
  A.~Farouk.
\newblock Explainable quantum clustering method to model medical data.
\newblock {\em Knowledge-Based Systems}, 267:110413, 2023.

\bibitem{saw1984chebyshev}
J.~G. Saw, M.~C. Yang, and T.~C. Mo.
\newblock Chebyshev inequality with estimated mean and variance.
\newblock {\em The American Statistician}, 38(2):130--132, 1984.

\bibitem{aumann2010some}
R.~J. Aumann.
\newblock {Some non-superadditive games, and their Shapley values, in the
  Talmud}.
\newblock {\em International Journal of Game Theory}, 39:1--10, 2010.

\bibitem{winter2002shapley}
E.~Winter.
\newblock {The Shapley value}.
\newblock {\em Handbook of game theory with economic applications},
  3:2025--2054, 2002.

\bibitem{hart1989shapley}
S.~Hart.
\newblock {Shapley value}.
\newblock In {\em Game theory}, pages 210--216. The New Palgrave. Palgrave
  Macmillan, London, 1989.

\bibitem{shapley1952value}
L.~S. Shapley.
\newblock {\em A Value for N-Person Games}.
\newblock RAND Corporation, Santa Monica, CA, 1952.

\bibitem{byrne2019counterfactuals}
R.~M. Byrne.
\newblock {Counterfactuals in Explainable Artificial Intelligence (XAI):
  Evidence from Human Reasoning}.
\newblock In {\em IJCAI}, pages 6276--6282, 2019.

\bibitem{guidotti2024counterfactual}
R.~Guidotti.
\newblock Counterfactual explanations and how to find them: literature review
  and benchmarking.
\newblock {\em Data Mining and Knowledge Discovery}, 38(5):2770--2824, 2024.

\bibitem{chen2023quantum}
G.~Chen, Q.~Chen, S.~Long, W.~Zhu, Z.~Yuan, and Y.~Wu.
\newblock Quantum convolutional neural network for image classification.
\newblock {\em Pattern Analysis and Applications}, 26(2):655--667, 2023.

\bibitem{montanaro2015quantum}
A.~Montanaro.
\newblock Quantum speedup of {Monte Carlo} methods.
\newblock {\em Proceedings of the Royal Society A: Mathematical, Physical and
  Engineering Sciences}, 471(2181):20150301, 2015.

\bibitem{ross2013elementary}
K.~A. Ross.
\newblock {\em Elementary analysis}.
\newblock Springer, New York, NY, 2013.

\bibitem{rattew2022preparing}
A.~G. Rattew and B.~Koczor.
\newblock Preparing arbitrary continuous functions in quantum registers with
  logarithmic complexity.
\newblock {\em arXiv preprint arXiv:2205.00519}, 2022.

\bibitem{plesch2011quantum}
M.~Plesch and {\v{C}}.~Brukner.
\newblock Quantum-state preparation with universal gate decompositions.
\newblock {\em Physical Review A}, 83(3):032302, 2011.

\bibitem{kay_2023}
A.~Kay.
\newblock Tutorial on the {Quantikz} package, arXiv: 1809.03842,
  \url{https://arxiv.org/abs/1809.03842}, March 2023.

\bibitem{giovannetti2008quantum}
V.~Giovannetti, S.~Lloyd, and L.~Maccone.
\newblock Quantum random access memory.
\newblock {\em Physical review letters}, 100(16):160501, 2008.

\bibitem{githubEntry}
I.~Burge, M.~Barbeau, and J.~Garcia-Alfaro.
\newblock {Quantum Algorithms for Shapley Value Calculation [extended
  simulation github repository]}, github,
  \url{https://github.com/iain-burge/QuantumShapleyValueAlgorithm}, May 2023.

\bibitem{grover1996fast}
L.~K. Grover.
\newblock A fast quantum mechanical algorithm for database search.
\newblock In {\em Proceedings of the twenty-eighth annual ACM symposium on
  Theory of computing}, pages 212--219, 1996.

\bibitem{burge2024quantum}
I.~Burge, M.~Barbeau, and J.~Garcia-Alfaro.
\newblock Quantum cordic--arcsin on a budget.
\newblock {\em arXiv preprint arXiv:2411.14434}, 2024.

\bibitem{ribeiro2016should}
M.~T. Ribeiro, S.~Singh, and C.~Guestrin.
\newblock " why should i trust you?" explaining the predictions of any
  classifier.
\newblock In {\em Proceedings of the 22nd ACM SIGKDD international conference
  on knowledge discovery and data mining}, pages 1135--1144, 2016.

\bibitem{bach2015pixel}
S.~Bach, A.~Binder, G.~Montavon, F.~Klauschen, K.-R. M{\"u}ller, and W.~Samek.
\newblock On pixel-wise explanations for non-linear classifier decisions by
  layer-wise relevance propagation.
\newblock {\em PloS one}, 10(7):e0130140, 2015.

\bibitem{marques2024explainability}
J.~Marques-Silva and X.~Huang.
\newblock Explainability is not a game.
\newblock {\em Communications of the ACM}, 67(7):66--75, 2024.

\bibitem{ghassemi2021false}
M.~Ghassemi, L.~Oakden-Rayner, and A.~L. Beam.
\newblock The false hope of current approaches to explainable artificial
  intelligence in health care.
\newblock {\em The Lancet Digital Health}, 3(11):e745--e750, 2021.

\bibitem{london2019artificial}
A.~J. London.
\newblock Artificial intelligence and black-box medical decisions: accuracy
  versus explainability.
\newblock {\em Hastings Center Report}, 49(1):15--21, 2019.

\bibitem{lipovetsky2023quantum}
S.~Lipovetsky.
\newblock Quantum-like data modeling in applied sciences.
\newblock {\em Stats}, 6(1):345--353, 2023.

\bibitem{van2022tractability}
G.~Van~den Broeck, A.~Lykov, M.~Schleich, and D.~Suciu.
\newblock {On the tractability of SHAP explanations}.
\newblock {\em Journal of Artificial Intelligence Research}, 74:851--886, 2022.

\bibitem{bertossi2020causality}
L.~Bertossi, J.~Li, M.~Schleich, D.~Suciu, and Z.~Vagena.
\newblock {Causality-based explanation of classification outcomes}.
\newblock In {\em Proceedings of the Fourth International Workshop on Data
  Management for End-to-End Machine Learning}, pages 1--10, 2020.

\bibitem{tarkowski2017game}
M.~K. Tarkowski, T.~P. Michalak, T.~Rahwan, and M.~Wooldridge.
\newblock Game-theoretic network centrality: A review.
\newblock {\em arXiv preprint arXiv:1801.00218}, 2017.

\bibitem{lloyd2014quantum}
S.~Lloyd, M.~Mohseni, and P.~Rebentrost.
\newblock Quantum principal component analysis.
\newblock {\em Nature physics}, 10(9):631--633, 2014.

\bibitem{rebentrost2014quantum}
P.~Rebentrost, M.~Mohseni, and S.~Lloyd.
\newblock Quantum support vector machine for big data classification.
\newblock {\em Physical review letters}, 113(13):130503, 2014.

\bibitem{beer2020training}
K.~Beer, D.~Bondarenko, T.~Farrelly, T.~J. Osborne, R.~Salzmann,
  D.~Scheiermann, and R.~Wolf.
\newblock Training deep quantum neural networks.
\newblock {\em Nature communications}, 11(1):808, 2020.

\bibitem{pira2024interpretability}
L.~Pira and C.~Ferrie.
\newblock On the interpretability of quantum neural networks.
\newblock {\em Quantum Machine Intelligence}, 6(2):52, 2024.

\bibitem{steinmuller2022explainable}
P.~Steinm{\"u}ller, T.~Schulz, F.~Graf, and D.~Herr.
\newblock e{X}plainable {AI} for {Quantum} {Machine} {Learning}.
\newblock {\em arXiv preprint arXiv:2211.01441}, 2022.

\bibitem{robbins1955remark}
H.~Robbins.
\newblock A remark on stirling's formula.
\newblock {\em The American mathematical monthly}, 62(1):26--29, 1955.

\end{thebibliography}

\appendix 
\renewcommand{\thesubsection}{\Alph{subsection}}
\subsection{Calculation for Alice}\label{appendix:byHandExample}

\noindent Quantum estimation of Alice's Shapley value by hand.
Let $\ell$ be equal to $2$.
Note that an auxiliary register stores the vote count.
To perform $U_V^-$, we begin with the state:
\begin{equation*}
    \ket{00}_\texttt{Pt}\ket{000}_\texttt{Pl}\ket{000}_\text{Aux}\ket{0}_\texttt{Ut}
\end{equation*}
The first qubit of the player register represents Alice, the second represents Bob, and the third represents Charlie.

\medskip

\noindent We perform the first step of the algorithm described in Section~\ref{sec:algorithm}, yielding:

\begin{align*}
    \sum\limits_{k=0}^{3} & \sqrt{w_2(k)} \ket{k}_\texttt{Pt} \Big[(1-t_2'(k))\ket{000}_\texttt{Pl}+\sqrt{t_2'(k)(1-t_2'(k))}\ket{001}_\texttt{Pl}\\
    &+\sqrt{t_2'(k)(1-t_2'(k))} \ket{010}_\texttt{Pl}+t_2'(k)\ket{011}_\texttt{Pl}\Big] \ket{000}_\text{Aux}\ket{0}_\texttt{Ut}
\end{align*}

\noindent Next, we tally the votes.
This step is the first half of the circuit in Figure \ref{fig:U_wOfVotingGame}, up to but not including the $\geq q$ gate.
We get,
\begin{align*}
    \sum\limits_{k=0}^{3} \sqrt{w_2(k)} \ket{k}_\texttt{Pt} 
    \Big[&
    (1-t_2'(k))\ket{000}_\texttt{Pl}\ket{000}_\text{Aux}
    +\sqrt{t_2'(k)(1-t_2'(k))}\ket{001}_\texttt{Pl}\ket{001}_\text{Aux}\\&
    +\sqrt{t_2'(k)(1-t_2'(k))} \ket{010}_\texttt{Pl}\ket{010}_\text{Aux}
    +t_2'(k)\ket{011}_\texttt{Pl}\ket{011}_\text{Aux}
    \Big] 
    \ket{0}_\texttt{Ut}.
\end{align*}

\noindent Performing the remainder of $U_V^-$ gives,
\begin{align*}
    \sum\limits_{k=0}^{3} \sqrt{w_2(k)} \ket{k}_\texttt{Pt} 
    \Big[&
    (1-t_2'(k))\ket{000}_\texttt{Pl}\ket{000}_\text{Aux}\ket{0}_\texttt{Ut}
    +\sqrt{t_2'(k)(1-t_2'(k))}\ket{001}_\texttt{Pl}\ket{001}_\text{Aux}\ket{0}_\texttt{Ut}\\&
    +\sqrt{t_2'(k)(1-t_2'(k))} \ket{010}_\texttt{Pl}\ket{010}_\text{Aux}\ket{0}_\texttt{Ut}
    +t_2'(k)\ket{011}_\texttt{Pl}\ket{011}_\text{Aux}\ket{0}_\texttt{Ut}
    \Big]. 
\end{align*}

\noindent Thus, the expected value when measuring the utility register, after performing $U_V^-$, is $\Phi^-_2(0) = 0$, where $\Phi^\pm_\ell(i)$ is the $\ell$th approximation of the $i$th player's Shapley value.

\medskip

\noindent To perform $U_V^+$, the exact same steps are performed except that we flip Alice to the on state,
\begin{equation*}
    \ket{00}_\texttt{Pt}\ket{100}_\texttt{Pl}\ket{000}_\text{Aux}\ket{0}_\texttt{Ut}.
\end{equation*}
where the first qubit of the player register represents Alice, the second represents Bob, and the third represents Charlie.

\medskip

\noindent We perform the first step of the algorithm from Section \ref{sec:algorithm}, yielding:
\begin{align*}
    \sum\limits_{k=0}^{3} & \sqrt{w_2(k)} \ket{k}_\texttt{Pt} \Big[(1-t_2'(k))\ket{100}_\texttt{Pl}+\sqrt{t_2'(k)(1-t_2'(k))}\ket{101}_\texttt{Pl}\\
    &+\sqrt{t_2'(k)(1-t_2'(k))} \ket{110}_\texttt{Pl}+t_2'(k)\ket{111}_\texttt{Pl}\Big] \ket{000}_\text{Aux}\ket{0}_\texttt{Ut}
\end{align*}

\noindent Next, we tally the votes.
This step is the first half of the circuit in Figure \ref{fig:U_wOfVotingGame}, up to but not including the "$\geq q$" gate.
We get,
\begin{align*}
    \sum\limits_{k=0}^{3} \sqrt{w_2(k)} \ket{k}_\texttt{Pt} 
    \Big[&
    (1-t_2'(k))\ket{100}_\texttt{Pl}\ket{011}_\text{Aux}
    +\sqrt{t_2'(k)(1-t_2'(k))}\ket{101}_\texttt{Pl}\ket{100}_\text{Aux}\\&
    +\sqrt{t_2'(k)(1-t_2'(k))} \ket{110}_\texttt{Pl}\ket{101}_\text{Aux}
    +t_2'(k)\ket{111}_\texttt{Pl}\ket{110}_\text{Aux}
    \Big] 
    \ket{0}_\texttt{Ut}.
\end{align*}

\noindent Performing the remainder of $U_V^+$ gives,
\begin{align*}
    \sum\limits_{k=0}^{3} \sqrt{w_2(k)} \ket{k}_\texttt{Pt} 
    \Big[&
    (1-t_2'(k))\ket{100}_\texttt{Pl}\ket{011}_\text{Aux}\ket{0}_\texttt{Ut}
    +\sqrt{t_2'(k)(1-t_2'(k))}\ket{101}_\texttt{Pl}\ket{100}_\text{Aux}\ket{1}_\texttt{Ut}\\&
    +\sqrt{t_2'(k)(1-t_2'(k))} \ket{110}_\texttt{Pl}\ket{101}_\text{Aux}\ket{1}_\texttt{Ut}
    +t_2'(k)\ket{111}_\texttt{Pl}\ket{110}_\text{Aux}\ket{1}_\texttt{Ut}
    \Big]. 
\end{align*}

\medskip

\noindent Thus, the expected value when measuring the utility register, after performing $U_V^+$, is $\Phi_2^+(0) \approx0.6617$.
Hence, the difference between the expected values $\Phi_2(0)=\Phi_2^+(0)-\Phi_2^-(0)$ is a close approximation for $\Phi(0) = 2/3 = 0.6666\cdots$.

\subsection{Equality of Modified Beta Function and Shapley Value Weights} \label{appendix:betaIsGamma}
\begin{lemma}
\label{lem:recurrence}
We have the following recurrence relationship:
\begin{equation*}
\beta_{n,0} = \beta_{n,n} = \frac{1}{n + 1}
\mbox{ and } 
\beta_{n,m} = \frac{m}{n-(m-1)}\beta_{n,m-1}.
\end{equation*}
\end{lemma}
\begin{proof}
There are two cases.

\noindent
Case~1 ($m$ is equal to zero, or $n$). We have the following integration.
\begin{equation*}
\beta_{n,0} = \int\limits_0^1 (1-x)^n dx
= - \frac{(1-x)^{n+1}}{n+1}\bigg\rvert_0^1
= \frac{1}{n+1}.
\end{equation*}
A nearly identical calculation can be used to show $\beta_{n,n}$ is equal to $\frac{1}{n+1}$.

\noindent
Case~2 ($0<m<n$). We have the following partial integration:
 \begin{align*}
\beta_{n, m} &= \int\limits_0^1 x^{m}(1-x)^{n-m} dx \\
        &= \frac{x^{m}(1-x)^{n-(m-1)}}{n-(m-1)}  \bigg\rvert_0^1 - \int\limits_0^1 \frac{-m}{n-(m-1)} x^{m-1} (1-x)^{n-(m-1)}dx\\
        &= 0 + \frac{m}{n-(m-1)} \int\limits_0^1 x^{m-1} (1-x)^{n-(m-1)}dx\\
        &= \frac{m}{n-(m-1)} 
        \beta_{n,m-1}.
 \end{align*}
\end{proof}

\begin{theorem}
The beta function $\beta_{n,m}$ is equal to the Shapley weight function $\gamma(n,m)$, with $0\leq m\leq n$ and $m,n\in \mathbb{N}$.
\end{theorem}
\begin{proof}
The proof is by induction on $m$.

\noindent
Base case ($m=0$). 
According to Case~1 of Lemma~\ref{lem:recurrence},
we have that $\beta_{n,0}$ is equal to $\frac{1}{n+1}$, which is equal to $\gamma(n,0)$.

\noindent
Inductive step ($m>0$). 
Suppose $\beta_{n,k}$ is equal to $\gamma(n,k)$, $k\in\mathbb{N}$, we need to show $\beta_{n,k+1}=\gamma(n,k+1)$, $0 \leq k < n$.
According to Case~2 of Lemma~\ref{lem:recurrence},
$\beta_{n,k+1}$ is equal to $\frac{k+1}{n-k}\beta_{n,k}$.
Using the inductive hypothesis, the latter is equivalent to
$$
\frac{k+1}{n-k}\gamma(n,k) = \frac{k+1}{n-k} \cdot \frac{k!(n-k)!}{(n+1)!}
$$
which matches the definition of $\gamma(n, k+1)$.
\end{proof}

\subsection{Riemann Sum Samples} \label{appendix:tprimeProof}
\begin{lemma}
    $t'_\ell(k)$ is in range $[t_\ell(k), t_\ell(k+1)]$.
\end{lemma}
\begin{proof}
    Recall the definition for $t_\ell(k)=\sin^2(\pi (k/2^{\ell+1}))$.
    Since $t'_\ell(k)=t_{\ell+1}(2k+1)$, we have that,
    \begin{equation*}
        t'_\ell(k) = \sin^2\left(\pi \frac{k+1/2}{2^{\ell+1}} \right) = t'_\ell\left(k+1/2\right).
    \end{equation*}
    Since $k \in \{ 0, \dots, 2^\ell-1 \}$, $t_\ell(x)$ is increasing for $x \in [k, k+1]$ for all $k$. 
    Therefore, since $k+1/2 \in [k, k+1]$, the result holds.
\end{proof}

\subsection{Complexity and Error Analysis}
\label{appendix:errorAnalysis}

\subsubsection{Error - Step 1}
It is important to assess the error and complexity of Step 1 of the Section~\ref{sec:algorithm} algorithm.

\begin{definition}[Darboux Sums]
    \label{definition:darboux}
    \emph{Darboux sums~\cite[page~270]{ross2013elementary}} takes a partition \mbox{$P=(z_0,z_1,\dots,z_r)$} of an interval $[a,b]$, where $a=z_0<z_1<\cdots<z_r=b$, and a function $f$ which maps $(a,b)$ to $\mathbb{R}$.
    Each interval $\left[ z_i, z_{i+1}\right]$ is called a \emph{subinterval}.
    Let 
    \begin{equation*}
        M_i = \sup\limits_{x\in[z_i,z_{i+1}]}f(x), \quad \mbox{and} 
    \end{equation*}
    
    \begin{equation*}
        m_i=\inf\limits_{x\in[z_i,z_{i+1}]} f(x), \quad i=0,\ldots,r-1.
    \end{equation*}
    The \emph{upper Darboux sum}
    and the \emph{lower Darboux sum} are:
    \begin{equation*}
        U(f,P) = \sum\limits_{i=0}^{r-1} M_i (z_{i+1}-z_i),
        \quad L(f,P) = \sum\limits_{i=0}^{r-1} m_i (z_{i+1}-z_i).
    \end{equation*}
\end{definition}

\begin{corollary} \label{corollary:riemannBound}
    Given a partition $P=(z_0,\dots,z_r)$ of an interval $[a,b]$ and an integrable function $f\colon [a,b] \to \mathbb{R}$, any Riemann sum (Definition~\ref{def:riemannSum}) is greater than or equal to the lower Darboux sum and less than or equal to than the upper Darboux sum,
    \begin{equation*}
        L(f,P) \leq \sum\limits_{i=0}^{r-1} f(z_i') (z_{i+1}-z_i) \leq U(f,P)
    \end{equation*}
    where $z_i'\in[z_i,z_{i+1}]$.
\end{corollary}

\begin{lemma} \label{lemma:integralBound}
    Given a partition $P=(z_0,\dots,z_r)$ of an interval $[a,b]$ and an integrable function $f\colon [a,b] \to \mathbb{R}$, the integral of $f$ from $a$ to $b$ is greater than the lower Darboux sum and smaller than the upper Darboux sum.
    \begin{equation*}
        L(f,P) \leq \int\limits_{a}^{b} f(x) dx \leq U(f,P).
    \end{equation*}
\end{lemma}
\begin{proof}
    By the Mean Value Theorem \cite[page~233]{ross2013elementary}, there exists an $p_i\in [z_i, z_{i+1}]$, such that,
    \begin{equation*}
        f(p_i) = \frac{1}{z_{i+1}-z_i} \int\limits_{z_i}^{z_{i+1}} f(t) dt.
    \end{equation*}
    Hence, we can create the following Riemann sum (Definition~\ref{def:riemannSum}),
    \begin{equation*}
        \sum\limits_{k=0}^{r-1} (z_{i+1}-z_i) f(p_i) = \sum\limits_{k=0}^{r-1} \int\limits_{z_i}^{z_{i+1}} f(x) dx = \int\limits_a^b f(x) dx.
    \end{equation*}
    Thus, by Corollary~\ref{corollary:riemannBound}, the result holds.
\end{proof}

\begin{lemma} \label{lemma:boxBounding}
    Let $P=(z_0,\dots,z_r)$ be a partition of an interval $[a,b]$ and suppose the function $f\colon [a,b] \to \mathbb{R}$ is integrable. 
    Then, the error of a Riemann sum is bounded by the difference of the upper and lower Darboux sums, that is, the difference between the Riemann sum and the integral is bounded by,
    \begin{equation*}
        \abs{\int\limits_a^b f(x) dx - \sum\limits_{k=0}^{r-1} f(z_k') (z_{k+1}-z_k)} \leq U(f, P) - L(f, P)
    \end{equation*}
    where $z_k' \in [z_k, z_{k+1}]$.
\end{lemma}
\begin{proof}
    We must show that the absolute difference between any given Riemann sum and the integral from $a$ to $b$ is less than or equal to the absolute difference between the upper and lower Darboux sums.
    Without loss of generality, suppose the integral is greater than the Riemann sum.
    By Lemma \ref{lemma:integralBound},
    \begin{equation*}
        \int\limits_a^b f(x) dx - \sum\limits_{k=0}^{r-1} f(z_k') (z_{k+1}-z_k) \leq U(f,P) - \sum\limits_{k=0}^{r-1} f(z_k') (z_{k+1}-z_k).
    \end{equation*}
    Similarly, by Corollary \ref{corollary:riemannBound},
    \begin{equation*}
        U(f,P) - \sum\limits_{k=0}^{r-1} f(z_k') (z_{k+1}-z_k) \leq U(f,P) - L(f,P).
    \end{equation*}
\end{proof}

Lemma~\ref{lemma:boxBounding} allows us to bound the error of any Riemann sum approximation for an integral, giving us a sequence of boxes that bound the error for their respective subinterval.
A convenient result of our method is that given only a partition; we can guarantee error bounds given any selection of samples.

\begin{corollary} \label{corollary:ErrorBound}
    Recall the function $b_{n,m}(x)=x^m(1-x)^{n-m}$ (Definition~\ref{def:betaFunction}) and partition $P_\ell=\left(t_\ell(0),\dots,t_\ell\left(2^\ell\right)\right)$ of $[0,1]$.
    Lemma \ref{lemma:boxBounding} implies that the error of a Riemann sum using $b_{n,m}$ and $P_\ell$ is upper bounded by,
    \begin{equation*}
        \sum\limits_{k=0}^{2^\ell-1} w_\ell(k) (M_k-m_k).
    \end{equation*}
    $w_\ell(k)$ is defined in Equation \eqref{eq:w} and $m_k,M_k$ are defined in Definition \ref{definition:darboux}.
\end{corollary}

It is important to find a simple upper bound for the function $w_\ell(k)$ and the difference $M_k-m_k$.

\begin{lemma} \label{lemma:wBound}
    $w_\ell(k)$ is upper bounded by a function that depends solely on $\ell$:
    \begin{equation*}
        w_\ell(k) \leq \frac{\pi}{2^{\ell+1}}.
    \end{equation*}
\end{lemma}
\begin{proof}
    Recall Equation~\eqref{eq:w}, $w_\ell(k)=\sin^2(a+b)-\sin^2(a)$, where $a=\pi k/2^{\ell+1}$, $b=\pi/2^{\ell+1}$.
    Then, we have,
    \begin{equation*}
        w_\ell(k) = [\sin(a+b)+\sin(a)][\sin(a+b)-\sin(a)].
    \end{equation*}
    Proceeding with common trigonometric identities yields,
    \begin{equation*}
        w_\ell(k) = \left[2\sin\frac{b}{2}\cos\frac{b}{2}\right]\left[2\sin{\frac{2a+b}{2}\cos\frac{2a+b}{2}}\right],
    \end{equation*}
    which is equal to,
    \begin{equation*}
        \sin(b)\sin(2a+b).
    \end{equation*}
    Since $0\leq \sin(2a+b)\leq 1$, and $\sin x\leq x$ when $x$ is non-negative, the result holds.

\end{proof}

\begin{remark} \label{remark:bSymmetry}
    It is helpful to note that $b_{n,m}(x)=b_{n,n-m}(1-x)$.
\end{remark}

\begin{lemma} \label{lemma:bDerivative}
    Let $n\in\mathbb{N}$, and suppose $m\in\{1,\cdots,n-1\}$, then the derivative of $b_{n,m}(x)$ is,
    \begin{equation*}
        \diff{b_{n,m}(x)}{x}  = \frac{(m-nx)b_{n,m}(x)}{x(1-x)} .
    \end{equation*}
\end{lemma}
\begin{proof}
    Recall the definition of $b_{n,m}(x)=x^m(1-x)^{n-m}$, then by the product rule,
    \begin{equation*}
        \diff{b_{n,m}(x)}{x} = mx^{m-1}(1-x)^{n-m} - (n-m)x^m(1-x)^{n-m-1}.
    \end{equation*}
    Pulling out the common terms yields,
    \begin{equation*}
        x^{m-1}(1-x)^{n-m-1}\left(m(1-x) - (n-m)x \right).
    \end{equation*}
    Collecting like terms and rearranging gives,
    \begin{equation*}
        \frac{(m-nx)( x^m(1-x)^{n-m})}{x(1-x)}.
    \end{equation*}
\end{proof}

\begin{lemma}
    For $m=0, n$, the derivatives of $b_{n,m}$ are,
    \begin{equation*}
        \diff{b_{n,0}(x)}{x} = n(1-x)^{n-1} \quad \text{and} \quad \diff{b_{n,n}(x)}{x} = nx^{n-1},
    \end{equation*}
    respectively.
\end{lemma}

\begin{lemma} \label{lemma:bnmMax}
    Let $n\in\mathbb{N}$ and $m\in\{0,\cdots,n\}$, then for all x in $[0,1]$,
    \begin{equation*}
        b_{n,m}(x) \leq b_{n,m}\left(m/n \right).
    \end{equation*}
\end{lemma}
\begin{proof}
    Suppose $m=n$, then $b_{n,m}(x)=x^n$.
    Clearly, $x^n$ is increasing on the interval $[0,1]$, thus $b_{n,m}$ is maximized at \mbox{$x=1=m/n$}.
    By Remark~\ref{remark:bSymmetry}, if $m=0$, the max of $b_{n,m}$ is also $1$ and this occurs at $x=0=m/n$.

    Finally, if $1\leq m\leq n-1$, then $b_{n,m}(0)=b_{n,m}(1)=0$.
    Note that $b_{n,m}(x)$ is strictly positive for $x\in(0,1)$, and there is only one critical point on that interval.
    Thus, the maximum must be the critical point $x$ between $(0,1)$.
    By Lemma~\ref{lemma:bDerivative}, this critical point satisfies $m-nx=0$.
    Therefore, $b_{n,m}(x)$ is maximized when $x=m/n$.
\end{proof}

\begin{corollary} \label{corollary:bnmGrowth}
    If $x<m/n$, then $b_{n, m}(x)$ is increasing at $x$,
    \begin{equation*}
        \diff{b_{n, m}(x)}{x} > 0,
    \end{equation*}
    and if $x>m/n$, then $b_{n, m}(x)$ is decreasing at $x$,
    \begin{equation*}
        \diff{b_{n, m}(x)}{x} < 0.
    \end{equation*}
\end{corollary}

\begin{definition}[Notation] \label{def:gammaEll}
    We define the function $\gamma_\ell(n,m)$ as the following Riemann sum approximation for $\gamma(n,m)$,
    \begin{equation*}
        \gamma_\ell(n, m) = \sum\limits_{k=0}^{2^\ell-1} w_\ell(k) \cdot b_{n,m}\left(t'_\ell(k)\right).
    \end{equation*}
    Let $u_\ell(n, m)$ be the signed difference between $\gamma_\ell(n,m)$ and $\gamma(n,m)$,
    \begin{equation*}
        u_\ell(n, m) = \gamma_\ell(n, m) - \gamma(n,m).
    \end{equation*}
\end{definition}

\begin{lemma} \label{lemma:uUpperBound}
    The estimation $\gamma_\ell(n,m)$ has an absolute error upper bounded by,
    \begin{equation*}
        \abs{u_\ell(n,m)} \leq \frac{\pi}{2^\ell} b_{n, m}(m/n).
    \end{equation*}
\end{lemma}

\begin{proof}
    Note, for the sake of space, we omit the cases where $m=0$ and $m=1$, though they do follow the same bound.
    Suppose $m\in \{1,\dots,n-1\}$.
    By Corollary~\ref{corollary:ErrorBound}, we have,
    \begin{equation*}
        \abs{u_\ell(n,m)} \leq \sum\limits_{k=0}^{2^\ell-1} w_\ell(k) (M_k-m_k).
    \end{equation*}
    Applying Lemma~\ref{lemma:wBound} shows this is less than or equal to,
    \begin{equation*}
        \frac{\pi}{2^{\ell+1}} \sum\limits_{k=0}^{2^\ell-1} M_k-m_k.
    \end{equation*}
    Define $j_k,J_k$ to be in interval $[t_\ell(k),t_\ell(k+1)]$ such that $m_k=b_{n,m}(j_k)$, $M_k=b_{n,m}(J_k)$.
    Then,
    \begin{equation*}
        \frac{\pi}{2^{\ell+1}} \sum\limits_{k=0}^{2^\ell-1} M_k-m_k = \frac{\pi}{2^{\ell+1}} \sum\limits_{k=0}^{2^\ell-1} b_{n,m}(J_k)-b_{n,m}(j_k).
    \end{equation*}

    For any $k$ such that $t_\ell(k+1)\leq m/n$, by Corollary~\ref{corollary:bnmGrowth}, $x\in(t_\ell(k), t_\ell(k+1))$ implies that $\frac{d}{dx}b_{n,m}(x)$ is strictly positive.
    Thus, the leftmost part of the subinterval $[t_\ell(k), t_\ell(k+1)]$ is the minimum with respect to $b_{n,m}(x)$ and the rightmost part is the maximum, so,
    \begin{equation} \label{eq:bnmGrowingDifference}
        b_{n,m}(J_k)-b_{n,m}(j_k) = b_{n,m}(t_\ell(k+1))-b_{n,m}(t_\ell(k)).
    \end{equation}
    Similarly, for $k$ such that $t_\ell(k)\geq m/n$, by Corollary~\ref{corollary:bnmGrowth}, $x\in(t_\ell(k), t_\ell(k+1))$ implies $\frac{d}{dx}b_{n,m}(x)$ is strictly negative, so,
    \begin{equation} \label{eq:bnmDecliningDifference} 
        b_{n,m}(J_k)-b_{n,m}(j_k) = -[b_{n,m}(t_\ell(k+1))-b_{n,m}(t_\ell(k))].
    \end{equation}
    If for some $k$, $t_\ell(k)$ is equal to $m/n$, then we have described each interval.
    Supposing that for all $k$, $t_\ell(k)\neq m/n$, there exists one other subinterval $[t_\ell(s), t_\ell(s+1)]$ that is non-monotonic.
    Its absolute error is bounded by,
    \begin{equation} \label{eq:bnmNonMonotonicDifference}
        b_{n,m}(J_s)-b_{n,m}(j_s) = \max\limits_{r\in\{0,1\}} b_{n,m}(m/n) - b_{n,m}(t_\ell(s + r)).
    \end{equation}

\noindent     Thus, we have the following upper bound for the absolute error,
    \begin{equation*} 
        \frac{\pi}{2^{\ell+1}} \left(\sum\limits_{k=0}^{s-1} (M_k - m_k)
        +(M_s-m_s) + \sum\limits_{k=s+1}^{2^\ell-1} (M_k - m_k)\right),
    \end{equation*}
    where we take $s=0$ when there is no non-monotonic sub-interval.
    Plugging in Equations~\eqref{eq:bnmGrowingDifference}, \ref{eq:bnmDecliningDifference}, and \ref{eq:bnmNonMonotonicDifference}, yields the absolute error upper bound of,
    \begin{scriptsize}
    \begin{align*} 
        \max\limits_{r\in\{0,1\}}\frac{\pi}{2^{\ell+1}} \left[
        \left(\sum\limits_{k=0}^{s-1} b_{n,m}(t_\ell(k+1))-b_{n,m}(t_\ell(k))\right)
        +\left(b_{n,m}(m/n) - b_{n,m}(t_\ell(s+r)) \right)
        -\left(\sum\limits_{k=s+1}^{2^\ell-1} b_{n,m}(t_\ell(k+1))-b_{n,m}(t_\ell(k))\right)\right].
    \end{align*}
    \end{scriptsize}
\noindent     We can telescope both series into,
    \begin{align*} 
        \max\limits_{r\in\{0,1\}} \frac{\pi}{2^{\ell+1}} \left[
        \left(b_{n,m}(t_\ell(s))-b_{n,m}(0) \right)
        + \left(b_{n,m}(m/n) - b_{n,m}(t_\ell(s+r)) \right)
        - \left(b_{n,m}\left(t_\ell\left(2^\ell\right)\right) - b_{n,m}(t_\ell(s+1)) \right)
        \right].
    \end{align*}
    Note that $b_{n,m}(0)=b_{n,m}\left(2^\ell\right)=0$, and that $-b_{n,m}(t_\ell(s+r))$ cancel out either $b_{n,m}(s)$ or $b_{n,m}(s+1)$.
    The above equation is upper bounded by,
    \begin{equation*}
        \max\limits_{r\in\{0,1\}} \frac{\pi}{2^{\ell+1}} \left(b_{n,m}(m/n) + b_{n,m}(t_\ell(s+r)) \right).
    \end{equation*}
    Thus, Lemma~\ref{lemma:bnmMax} implies the result.
\end{proof}

\begin{lemma}\label{lemma:nCmTimesbnm}
    Let $n$ be greater or equal to two and suppose $m$ is in $\{1, \dots, n-1\}$, we have the following inequality,
    \begin{equation*}
        {n\choose m} \cdot b_{n, m}(m/n) < \frac{1}{\sqrt{2\pi}} \cdot \sqrt{\frac{n}{m(n-m)}}.
    \end{equation*}
\end{lemma}
\begin{proof}
    First note that $b_{n,m}(m/n)$ can be rewritten as,
    \begin{equation} \label{eq:bnmAtmOvern}
        \frac{m^m (n-m)^{n-m}}{n^n}. 
    \end{equation}
    Additionally, by a slightly improved version of Stirling's approximation \cite{robbins1955remark}, for $k$ greater than or equal to one,
    \begin{equation} \label{eq:stirlingBounds}
        \sqrt{2\pi k} \left(\frac{k}{e}\right)^k e^\frac1{12k+1} 
        < k! 
        < \sqrt{2\pi k} \left(\frac{k}{e}\right)^k e^\frac1{12k}.
    \end{equation}
    It follows directly from Equation \eqref{eq:stirlingBounds}, and the inequality \hbox{$(12n)^{-1} - (12m+1)^{-1} - (12(n-m)+1)^{-1} < 0$}, that,
    \begin{equation} \label{eq:nChoosemUpperBound}
        {n \choose m} = 
        \frac{n!}{m!(n-m)!} < 
        \frac{1}{\sqrt{2\pi}} \cdot \sqrt{\frac{n}{m(n-m)}} \cdot \frac{n^n}{m^m(n-m)^{n-m}}.
    \end{equation}
    Multiplying the Equations~\eqref{eq:bnmAtmOvern} and~\eqref{eq:nChoosemUpperBound} yields the result.
\end{proof}

\begin{definition}\label{def:ShapleyEst}
    We write a Shapley value estimation using $\gamma_\ell$ as,
    \begin{equation*}
        \Phi_\ell(i) = \sum\limits_{S \subseteq F \setminus \{i\}} \gamma_\ell(\abs{S}, \abs{F\setminus\{i\}}) \cdot (V(S\cup\{i\})-V(S)).
    \end{equation*}
\end{definition}

\begin{lemma}\label{lemma:upperBoundForSumChoosenmTimesBnm}
    The following inequality holds for $n$ greater or equal to three,
    \begin{equation*}
        \sum\limits_{m=1}^{n-1} {n \choose m} b_{n, m}(m/n) \leq \sqrt{\frac{\pi n}{2}}.
    \end{equation*}
\end{lemma}
\begin{proof}
    The bound for odd $n$ is similar to finding a bound for an even $n$. 
    We proceed assuming $n$ is equal to $2k+1$ where $k\in \mathbb{N}$.
    By Lemma~\ref{lemma:nCmTimesbnm}, we immediately get the upper bound, 
    \begin{equation*}
        \sqrt\frac{n}{2\pi} \sum\limits^{2k}_{m=1} \frac{1}{\sqrt{m(n-m)}}.
    \end{equation*}
    This can be rewritten as,
    \begin{equation*}
        \sqrt\frac{n}{2\pi}\left[
        \left(\sum\limits^{k-1}_{m=1} \frac{1}{\sqrt{m(n-m)}}\right)
        + \left(\frac{1}{\sqrt{k(n-k)}} \right)
        + \left(\sum\limits^{2k}_{m=k+1} \frac{1}{\sqrt{m(n-m)}} \right)
        \right].
    \end{equation*}
    The expansions of the summations are equal. 
    We have the upper bound,
    \begin{equation}\label{eq:nCmTimesbnmUppBound1}
        \sqrt\frac{2n}{\pi}\left(\frac{1}{2\sqrt{k(n-k)}} + \sum\limits^{k-1}_{m=1} \frac{1}{\sqrt{m(n-m)}} \right).
    \end{equation}

\noindent     Next, note that $1/\sqrt{m(2k-m)}$ is a decreasing function on the interval $(0,n/2]$.
    So,
    \begin{equation}\label{eq:nCmTimesbnmLowerSum}
        \frac{1}{2\sqrt{k(n-k)}}+\sum^{k}_{m=1} \frac{1}{\sqrt{m(2k-m)}},
    \end{equation}
    can be interpreted as the lower Darboux sum of function $(m(n-m))^{-1/2}$,
    on the interval $(0,k]$, using partition \hbox{$P=\left(0,1,\dots,k-1,k,k+1/2\right)$}.

\noindent     By Lemma~\ref{lemma:integralBound}, this implies Equation~\eqref{eq:nCmTimesbnmLowerSum} is upper bounded by,
    \begin{equation*}
        \int\limits^{n/2}_0 \frac{1}{\sqrt{x(n-x)}}dx.
    \end{equation*}
    By converting the product $x(n-x)$ to vertex form, performing a substitution, and finally performing a trigonometric substitution, we find the definite integral equals $\pi/2$.
    Combining this with Equation~\eqref{eq:nCmTimesbnmUppBound1} shows the result.
\end{proof}

\begin{theorem}\label{theorem:ShapleyError}
    If $n$ is greater or equal to two, assuming Steps~2 and~3 introduce no error, the Shapley value estimation has the following upper bound on absolute error,
    \begin{equation}
        \abs{\Phi_\ell(i) - \Phi(i)} \leq \frac{(V_{\max}-V_{\min}) \sqrt{n}}{2^{\ell-3}}.
    \end{equation}
\end{theorem}
\begin{proof}
    By Definitions~\ref{def:gammaEll} and~\ref{def:ShapleyEst}, the $\ell\text{th}$ Shapley value approximation can be written as,
    \begin{equation*}
        \Phi_\ell(i) = \sum\limits_{S \subseteq F \setminus \{i\}} [\gamma(\abs{S}, \abs{F\setminus\{i\}}) + u_\ell(\abs{S}, \abs{F\setminus\{i\}})] \cdot [V(S\cup\{i\})-V(S)].
    \end{equation*}
    Pulling the terms apart and applying the standard equation for Shapley values give,
    \begin{equation*}
        \Phi_\ell(i) = \Phi(i) + \sum\limits_{S \subseteq F \setminus \{i\}} u_\ell(\abs{S}, \abs{F\setminus\{i\}}) \cdot (V(S\cup\{i\})-V(S)).
    \end{equation*}
    Thus, the Shapley estimation absolute error is equal to,
    \begin{equation*}
        \abs{\Phi_\ell(i) - \Phi(i)} = \abs{\sum\limits_{S \subseteq F \setminus \{i\}} u_\ell(\abs{S}, \abs{F\setminus\{i\}}) \cdot (V(S\cup\{i\})-V(S))}.
    \end{equation*}
    which is thus upper bounded by,
    \begin{equation*}
        \sum\limits_{S \subseteq F \setminus \{i\}} \abs{u_\ell(\abs{S}, \abs{F\setminus\{i\}})} \cdot \abs{(V(S\cup\{i\})-V(S))}.
    \end{equation*}
    Recall $V_{\max}$ and $V_{\min}$, defined in Equations~\eqref{eq:Vmax} and~\eqref{eq:Vmin}.
    This gives a new upper bounds for $\abs{\Phi_\ell(i)-\Phi(i)}$,
    \begin{equation*}
        (V_{\max}-V_{\min}) \cdot \sum\limits_{S \subseteq F \setminus \{i\}} \abs{u_\ell(\abs{S}, \abs{F\setminus\{i\}})}.
    \end{equation*}
    By Lemma~\ref{lemma:uUpperBound} and a reinterpretation of the sum, we have the upper bound,
    \begin{equation*}
        \frac{\pi (V_{\max}-V_{\min})}{2^\ell} \cdot \sum\limits_{m=0}^{n} {n\choose m} b_{n, m}(m/n).
    \end{equation*}

    For $n$ equal to two, the sum is less than $(V_{\max}-V_{\min})/2^{\ell-3}$. Hence, the result holds.
    Assuming $n$ greater or equal to three, by Lemma~\ref{lemma:upperBoundForSumChoosenmTimesBnm}, the above is less than,
    \begin{equation*}
        \frac{\pi (V_{\max}-V_{\min})}{2^\ell} \left(\sqrt{\frac{\pi n}{2}} + {n \choose 0} 1 + {n \choose n} 1 \right) = \frac{\pi (V_{\max}-V_{\min})}{2^\ell} \left(\sqrt{\frac{\pi n}{2}} + 2 \right).
    \end{equation*}
    Hence, the result holds.
\end{proof}

\subsubsection{Error - Step 2}
Suppose $U_V^\pm$ is an imperfect implementation of $\hat V^\pm$, such that,
\begin{align} \label{eq:step2Delta}
    \delta_h^\pm = \braket{V^\pm(h)}{1}\braket{1}{V^\pm(h)} - \hat V^\pm(h),\quad \delta_h = \delta_h^+ - \delta_h^-
\end{align}
where $\ket{V^\pm(h)}=U_V^\pm(h) \ket{0}$ (Equation~\eqref{eq:U_V}).
We also define,
\begin{equation} \label{eq:delta_max}
    \delta_{\max} \geq \max\limits_{h\in \{0,\dots, 2^n - 1\}} \abs{\delta_h}.
\end{equation}
In plain language, $\delta_h$ is the error of our quantum implementation in how much player $i$ would contribute to the coalition represented by the binary string $h$.
Additionally, $\delta_{\max}$ upper bounds the error of $\delta_h$.

With these definitions in mind, let us see how the error in our quantum implementation propagates.
Let us first recall our state after Step~2 (Section~\ref{sec:algorithm}).
\begin{equation*}
    \tr_{\texttt{Pt},\texttt{Pl}}\left(\ketbra{\psi_2^\pm}{\psi_2^\pm}\right)=
    \sum\limits_{m=0}^n \sum\limits_{h\in H_m} \left(\sum\limits_{k=0}^{2^\ell-1} w_\ell(k) b_{n,m}\left(t_\ell'(k)\right) \right) \cdot \ket{V^\pm(h)}_\texttt{Ut} \bra{V^\pm(h)}_\texttt{Ut}.
\end{equation*}
By Definition~\ref{def:gammaEll}, this is equal to,
\begin{equation*}
    \sum\limits_{m=0}^n \sum\limits_{h\in H_m} \left( \gamma(n, m) + u_\ell(n, m) \right)\ket{V^\pm(h)}_\texttt{Ut} \bra{V^\pm(h)}_\texttt{Ut}.
\end{equation*}

\noindent We are interested in the expected value of the \texttt{Ut} register, which is equal to,
\begin{equation*}
    \bra{\psi_2^\pm} (I^{\otimes \ell+n} \otimes \ketbra{1}) \ket{\psi_2^\pm}=\sum\limits_{m=0}^n \sum\limits_{h\in H_m} \left( \gamma(n, m) + u_\ell(n, m) \right) \braket{V^\pm(h)}{1}\braket{1}{V^\pm(h)}.
\end{equation*}
By Equation~\eqref{eq:step2Delta}, this is equivalent to,
\begin{equation*}
    \sum\limits_{m=0}^n \sum\limits_{h\in H_m} \left( \gamma(n, m) + u_\ell(n, m) \right) \left( \hat V^\pm(h) + \delta_h^\pm \right).
\end{equation*}
Subtracting our estimates for $\Phi_i^-$ from $\Phi_i^+$ and multiplying by $V_{\max}-V_{\min}$ results in, 
\begin{equation*}
    (V_{\max}-V_{\min}) \left(\sum\limits_{m=0}^n \sum\limits_{h\in H_m} \left( \gamma(n, m) + u_\ell(n, m) \right) \left( \left(\hat V^+(h) - \hat V^-(h) \right) + \delta_h \right)\right).
\end{equation*}
Multiplying the terms out yields,
\begin{equation*}
    \Phi_i
    + (V_{\max}-V_{\min}) \left(
    \sum\limits_{m=0}^n \sum\limits_{h\in H_m} \left[\gamma(n, m) \delta_h \right]
    + \sum\limits_{m=0}^n \sum\limits_{h\in H_m} \left[u_\ell(n, m) \left( \hat V^+(h) - \hat V^-(h) \right) \right]
    + \sum\limits_{m=0}^n \sum\limits_{h\in H_m} \left[u_\ell(n, m) \delta_h \right]
    \right).
\end{equation*}

\noindent Thus, we have three error terms, and we now find an absolute bound for each of them.
The first term can be bounded easily, using Lemma~\ref{lemma:gammaSumsToOne} and Equation~\eqref{eq:step2Delta}, we have the following,
\begin{equation*}
    \abs{\sum\limits_{m=0}^n \sum\limits_{h\in H_m} \gamma(n, m) \delta_h }
    \leq \sum\limits_{m=0}^n \sum\limits_{h\in H_m} \gamma(n, m) \abs{\delta_h}
    \leq \delta_{\max} \sum\limits_{m=0}^n \sum\limits_{h\in H_m} \gamma(n, m)
    = \delta_{\max}.
\end{equation*}
The second term is calculated in Theorem~\ref{theorem:ShapleyError}.
In particular, noting that we have $\hat V^\pm(h) \in [0,1]$,
\begin{equation*}
    \abs{\sum\limits_{m=0}^n \sum\limits_{h\in H_m} u_\ell(n, m) \left( \hat V^+(h) - \hat V^-(h) \right)}
    \leq \frac{\sqrt{n}}{2^{\ell-3}}.
\end{equation*}
Finally, we have the error introduced by the third term.
By an almost identical calculation to that in Theorem~\ref{theorem:ShapleyError}, we find that,
\begin{equation*}
    \abs{\sum\limits_{m=0}^n \sum\limits_{h\in H_m} u_\ell(n, m) \delta_h}
    \leq \frac{\delta_{\max} \sqrt{n}}{2^{\ell-3}}.
\end{equation*}

\noindent Thus, the combined error introduced from Steps~1 and~2 is bounded by:
\begin{equation} \label{eq:errorStep1&2}
    \left( V_{\max} - V_{\min}\right) \cdot
    \left(
    \frac{\sqrt{n}}{2^{\ell-3}} \delta_{\max} + \frac{\sqrt{n}}{2^{\ell-3}}  + \delta_{\max}
    \right).
\end{equation}
Note that, no matter how large we make $\ell$, the error persists if we are unable to also minimize $\delta_{\max}$.

\subsubsection{Error - Step~3}
For this step, we use a quantum subroutine for speeding up Monte Carlo Methods \cite{montanaro2015quantum}.
With a probability greater than $8/\pi^2$, the procedure succeeds.
The process can be repeated $\mathcal{O}(\log(1/p))$ times to improve the success probability to $1-p$.
Suppose we perform $t$ iterations of amplitude estimation.
If we are trying to extract value $\mu$, and the result is $\tilde\mu$, the error is bounded,
\begin{equation*}
    \abs{\tilde \mu-\mu} \leq 2\pi \frac{\sqrt{\mu(1-\mu)}}{t} + \frac{\pi^2}{t^2}.
\end{equation*}
In our case, $\mu = (\Phi_i-V_{\min})/(V_{\max}-V_{\min})\in [0,1]$, so we can take the following more simple upper bound,
\begin{equation*}
    \abs{\tilde \mu-\mu} \leq 2\pi\frac{t\sqrt{\mu}+\pi}{t^2}.
\end{equation*}

\noindent Since we are extracting the value $\mu=(V_{\max}-V_{\min})\bra{\psi_2} (I^{\otimes \ell + n} \otimes \ketbra{1}) \ket{\psi_2}$, 
our final error is,
\begin{equation*}
    \left( V_{\max} - V_{\min} \right)
    \left(
    \frac{\sqrt{n}}{2^{\ell-3}} \delta_{\max} + \frac{\sqrt{n}}{2^{\ell-3}}  + \delta_{\max} + 2\pi\frac{t\sqrt{\mu}+\pi}{t^2}
    \right).
\end{equation*}

\subsubsection{Complexity}
Suppose we want a total error of $\epsilon$.
This section finds an upper bound of complexity given a desired error.
Let $\sqrt{n}/2^{\ell-3}, \delta_{\max}, 2\pi(t\sqrt{\mu}+\pi)/{t^2} \leq \epsilon'$ be less than or equal to $1$, then the error is bounded by,
\begin{equation*}
    \left( V_{\max} - V_{\min} \right)
    \left( (\epsilon')^2 + 3\epsilon' \right) 
    \leq 
    4 \left( V_{\max} - V_{\min} \right) \epsilon'.
\end{equation*}
Hence, an upper bound of error $\epsilon$ can be achieved if, $\epsilon' \leq \epsilon/(4\left( V_{\max} - V_{\min} \right))$.

\medskip

\noindent For notational simplicity, we define,
\begin{equation}\label{eq:lambda}
    \lambda = \frac{V_\text{max} - V_\text{min}}{\epsilon}.
\end{equation}
to get $\sqrt{n}/2^{\ell-3}$ less than $\epsilon'$, a partition register of size,
\begin{equation} \label{eq:ellSize}
    \ell = \left\lceil \log_2 (\lambda\sqrt{n}) \right\rceil + 5,
\end{equation}
is sufficient.

\medskip

\noindent Recall $C_D(\ell)$ to be the complexity of implementing $D_\ell$ (cf. Definition~\ref{def:C_D}) with an error with respect to $w_\ell(k)$ of less than $4^{-\ell}$.
$\delta_{\max}$ on the other hand depends on the complexity of computing $U_V^\pm$ with max error $\epsilon'$, which we write $C_V(\epsilon')$.
\begin{definition} \label{def:C_V}
    We denote the complexity to implement $U_V^\pm$, such that the maximum error $\delta_{\max}$ (Equations~\eqref{eq:step2Delta} and~\eqref{eq:delta_max}) is upper bounded by $\epsilon'$, as $C_V(\epsilon')$.
\end{definition}
Additionally, to sufficiently bound the error introduced by Step~3, we must satisfy \hbox{$2\pi(t\sqrt{\mu}+\pi)/{t^2} \leq \epsilon/(4(V_{\max}-V_{\min}))$}.
Suppose that we wish our absolute error $\epsilon$ to be less than the difference $\Phi_i-V_{\min}$.
Using the quadratic formula, we see it is sufficient to take $t$ equal to,
\begin{equation} \label{eq:tsize}
    t = \left\lceil \frac{28\sqrt{(\Phi_i-V_{\min})(V_{\max}-V_{\min})}}{\epsilon} \right\rceil.
\end{equation}

Finally, with the upper bounds for these important factors in mind, let us consider the complexity in terms of CNOT gates given these upper bounds.
Generating $\ket{\psi_{1a}}$ is complexity $C_D(\ell)$.
Next, evolving to $\ket{\psi_{1b}}$ takes $n\ell$ controlled rotations, which can be implemented with $2n\ell$ CNOTs in $2\cdot\text{max}(n,\ell)$ layers.
From here, producing state $\ket{\psi_2}$ has complexity $C_V(4^{-1}\lambda^{-1})$.
Next, Step~3 involves amplitude estimation, meaning the above steps are repeated $t$ times, followed by an Inverse Quantum Fourier Transform on $\log_2(t)$ bits taking $\mathcal{O}(\log_2^2t)$ operation.
Hence, the total complexity is bounded by,
\begin{equation} \label{eq:finalComplexity}
    t\cdot \left(C_D(\ell)+2n\ell + C_V\left(\frac{1}{4\lambda}\right)\right) + \mathcal{O}(\log^2t),
\end{equation}
where $\lambda$ is defined in Equation~\eqref{eq:lambda}, and conservative values for $\ell$ and $t$ are given in Equations~\eqref{eq:ellSize} and~\ref{eq:tsize} respectively.

\begin{theorem} \label{theorem:totalComplexity}
    Suppose there exists a constant real positive number $b$ such that $C_V(a\epsilon) > (b/a) C_V(\epsilon)$ for all real positive $a$. 
    The total complexity of the algorithm is,
    \begin{equation*}
        \mathcal{O} \left( \frac{\sqrt{(\Phi_i-V_{\min})(V_{\max}-V_{\min})}}{\epsilon} \left(C_D(\log_2(\lambda \sqrt{n})) + n \log_2(\lambda\sqrt{n}) + C_V\left(\frac{1}{4\lambda} \right) \right) \right).
    \end{equation*}
\end{theorem}
\begin{proof}
   The result follows directly from Equation~\eqref{eq:finalComplexity}.
\end{proof}

\balance 

\end{document}